\documentclass[11pt]{article}
\usepackage[table]{xcolor}
\usepackage{booktabs}
\usepackage{tabularx}
\usepackage{longtable}
\usepackage{array}
\usepackage{float}
\usepackage[T1]{fontenc}
\usepackage[utf8]{inputenc}
\usepackage{comment}
\usepackage{natbib}
\usepackage{graphicx}
\usepackage{amssymb}
\usepackage{orcidlink}

\usepackage{amsmath,amsthm,mathtools}
\newtheorem{theorem}{Theorem}
\newtheorem{lemma}{Lemma}
\newtheorem{proposition}{Proposition}
\theoremstyle{remark}

\usepackage{threeparttable}
\usepackage{siunitx}
\usepackage[ruled,vlined]{algorithm2e}

\usepackage{JABES_manu}
\usepackage{marvosym} 
\usepackage{hyperref} 
\graphicspath{{FIGURES/}}

\begin{document}

\title{Bayesian Analysis of Covariate-Driven Hawkes Processes with Application in Plant Epidemiology}

\author{
    Katarzyna Adamczyk-Chauvat$^{1,*}$ \orcidlink{0000-0001-7953-9153}\and Madalina Deaconu$^2$ \and Sylwester Masny$^3$ \orcidlink{0000-0002-0052-7625} \and Guillaume Kon Kam King$^1$
    \thanks{\\
        $^1$MaIAGE Research Unit, INRAE, Domaine de Vilvert, Université Paris‑Saclay, 78352 Jouy‑en‑Josas Cedex, France.\\
        $^2$Université de Lorraine, CNRS, Inria, IECL, F-54000 Nancy, France. \\
        $^3$Plant Protection Department,  National Institute of Horticultural Research, Skierniewice, Poland. \\
        $^*$Corresponding author: Katarzyna.Adamczyk@inrae.fr
    }
}

\date{}
\makeatletter
\renewcommand{\@fnsymbol}[1]{}
\makeatother
\maketitle

\begin{center}
    \textbf{Abstract}
\end{center}

Hawkes processes are widely used to model temporal point patterns
exhibiting self-exciting behavior. In this work, we consider a Hawkes
model with a baseline intensity depending on a set of dynamic
covariates and an exponential triggering function. A Bayesian approach
is applied to variable selection in this context. Spike-and-slab
priors are proposed to enable a parsimonious choice of relevant
variables. We also provide practical simulation procedures for the
fitted covariate-driven Hawkes model, which support posterior
predictive checks and probabilistic forecasting.  This approach is
illustrated by the example of spore release of \textit{Venturia
  inaequalis}, the fungus responsible for apple scab disease. We study
the dependence of primary spore releases on environmental covariates
and highlight the self-exciting behavior of subsequent releases. The
Bayesian framework also provides a natural basis to forecast future
spore releases.  By delivering interpretable posterior inclusion
probabilities alongside uncertainty-aware predictive risk bands, this
new methodology provides a rigorous foundation for managing temporal
event risks driven by both external environmental forcing and internal
contagion.

\vspace*{.1in}
\noindent\textsc{Keywords}: {Hawkes processes; Bayesian variable selection; spike-and-slab priors; spore release; apple scab}

\newpage

\section{Introduction}
\label{sec:Int}

Hawkes processes are used to model temporal point patterns exhibiting
self-exciting behavior: an earthquake followed by a sequence of
aftershocks; a post on a social network generating tweets or shares;
an infected host spreading a disease, etc.  They have become a popular
tool for numerous applications, including seismology
\citep{hawkes1,ogata}, finance for high-frequency trading and market
activity \citep{bacry}, social networks for information diffusion
\citep{wang}, and epidemiology for modelling the spread of infectious
diseases \citep{schoenberg,park}.  Additional examples of Hawkes
process applications and extensions are discussed in recent review
papers \citep{bernabeu,laub}.

The Hawkes process is defined as a Poisson cluster process
\citep{Hawkes2}, with independent clusters generated around center
points according to a branching structure.  The cluster centers,
referred to as \textit{immigrants}, are assumed to follow a Poisson
process with intensity $\lambda(t)$, meaning that the immigrants occur
spontaneously.  In practice, this assumption is often unrealistic: the
points that initiate the triggering process are triggered themselves
by exogenous covariates.  For instance, in the plant epidemiology
example which motivated this work, releases of pathogen fungus spores
are driven by meteorological conditions such as leaf wetness,
humidity, pressure and light.  In this paper, we consider the temporal
Hawkes process with a baseline intensity depending on covariates
through an appropriate link function. The model extends the
inhomogeneous Poisson process with covariate-dependent intensity. The
selection of relevant covariates is the main focus here, and a
Bayesian framework is adopted to address this question.

Several strands of literature are related to this problem.  Bayesian
inference for Hawkes processes was developed by
\citet{rasmussenBayesianInferenceHawkes2013} for parametric models,
using Markov chain Monte Carlo methods to approximate the posterior
distribution of the parameters.  Nonparametric Bayesian inference for
multivariate Hawkes processes was subsequently studied by
\citet{donnetNonparametricBayesianEstimation2020, sulem}, with a focus
on posterior concentration for the interaction functions. Bayesian
methods have also been developed for nonlinear Hawkes processes,
including variational Bayesian inference with Gaussian-process
self-effects \citep{malem-shinitskiVariationalBayesianInference2022},
and recent computational advances based on Integrated Nested Laplace
Approximation have been proposed for Epidemic-Type Aftershock Sequence
models and related Hawkes formulations \citep{serafini2023}.  Other
recent contributions address incomplete temporal information,
including count data, binned observations and interval-censored Hawkes
processes \citep{cheyssonSpectralEstimationHawkes2022,
  shlomovichParameterEstimationBinned2022,
  rizoiuIntervalcensoredHawkesProcesses2022}.  These works provide
important tools for Bayesian or likelihood-based inference in Hawkes
models, but they do not directly address Bayesian variable selection
for exogenous covariates entering the baseline intensity.  Hawkes and
self-exciting point process models with covariates have also been
considered in several application domains. For example,
\citet{reinhartSelfExcitingPointProcesses2018} proposed a
spatio-temporal self-exciting point process incorporating spatial
covariates for crime data. In neuroscience,
\citet{fallahiMultivariateGeneralizedLinear2023} introduced a
multivariate generalized linear Hawkes process with covariates in the
intensity function and considered simultaneous estimation and variable
selection in a high-dimensional regime using penalised
estimation. More recently,
\citet{perrinHawkesProcessDiffusiondriven2025} studied a Hawkes
process whose baseline intensity is driven by an external diffusion
process, establishing probabilistic properties and maximum-likelihood
inference results. These contributions demonstrate the importance of
external covariates in self-exciting point process models.  However,
variable selection in this literature has mainly been addressed
through penalised frequentist methods, or through model selection for
sparse interaction structures in multivariate Hawkes processes.  By
contrast, Bayesian variable selection for covariates entering the
baseline intensity remains comparatively unexplored.  A Bayesian
formulation is particularly useful in this setting because it provides
posterior inclusion probabilities for the covariates and allows
parameter uncertainty and selection uncertainty to be propagated
naturally into probabilistic forecasts.

In this paper, we consider a univariate Hawkes process with a baseline
intensity depending on dynamic covariates through a log-linear link
function. The model extends the inhomogeneous Poisson process with
covariate-dependent intensity by adding a self-exciting component. The
main inferential objective is to identify which covariates affect the
baseline rate while simultaneously estimating the strength and time
scale of self-excitation. To this end, we develop a Bayesian
variable-selection framework based on spike-and-slab priors on the
regression coefficients. This provides posterior inclusion
probabilities for the covariates, regularised estimation of their
effects, and uncertainty quantification for both the exogenous and
self-exciting components of the model.

The main contributions of this work are threefold. First, we introduce
a covariate-driven Hawkes model in which the baseline intensity is a
log-linear function of dynamic exogenous covariates, while the
triggering component is represented by an exponential kernel.  Second,
we develop a Bayesian variable-selection framework for this model
using spike-and-slab priors, allowing covariate relevance to be
assessed through posterior inclusion probabilities. Third, we provide
exact simulation procedures for the fitted covariate-driven Hawkes
process, including an explicit implementation of Ogata's
modified-thinning construction for the covariate-dependent baseline
and conditional forecasting setting considered here, as well as a
conditional cluster representation (detailed with proofs of validity
and complexity analysis in Appendix~\ref{sec:supp-simulation}). These
simulation tools are used for posterior predictive checking and
probabilistic forecasting, thereby propagating posterior uncertainty
in the model parameters.

The proposed methodology is applied to ascospore release data for
\textit{Venturia inaequalis}, the fungus responsible for apple scab
disease. Accurate probabilistic forecasts of ascospore release risk
are important for disease management, because they can support
decisions on fungicide applications and help reduce unnecessary
treatments. The Bayesian framework developed here provides both a
parsimonious selection of relevant meteorological variables and a
probabilistic description of future release risk conditional on
observed or forecast covariate trajectories.

The remainder of the paper is organised as follows. Section~2 recalls
the definition of the Hawkes process and introduces the
covariate-driven extension. Section~3 presents the likelihood, prior
distributions and posterior sampling strategy, including the
spike-and-slab variable-selection layer. Section~4 provides a
simulation study illustrating posterior concentration and covariate
selection. Section~5 applies the model to apple scab risk assessment
and presents posterior predictive checks and probabilistic forecasting
experiments. Section~6 concludes with a discussion of possible avenues
for improvement and potential generalizations. Methodological
derivations, exact simulation algorithms with proofs of validity, and
MCMC convergence diagnostics are provided in the Appendix.


\section{Hawkes process}
\label{sec:hawkes}
\subsection{Temporal point process}

A temporal point process models the times of occurrences of a random
event, represented by the points $\{t_i\}$ on a half-line.  We can
define it as a counting process $\{N (t), t \geq 0\}$ where $N(t)$
represents the total number of events up to time $t$. Consequently,
$N(t)$ satisfies the following conditions:
\begin{itemize}
    \item[(i)] $N(t) \geq 0$;
    \item[(ii)] $N(t)$ is integer-valued;
    \item[(iii)] If $s<t$ then
          $N(s)\leq N(t)$;
    \item[(iv)] for $s<t$, $N(t)-N(s)$ equals the number of
          events that have occurred in the time interval $(s,t]$.
\end{itemize}

As demonstrated in \citep{daley}, the probability structure of the
process $\{N (t), t \geq 0\}$ is uniquely determined by the
\textit{conditional intensity function}, defined as follows:
\begin{equation}
    \label{condInt}
    \lambda^{\star}(t) = \lim_{dt \rightarrow 0^{+}}
    \frac{\mathbb{P}\Bigl(N(t+dt)-N(t)=1 \;\big|\; \mathcal{F}_{t^-}\Bigr)}{dt},
\end{equation}
where the filtration $\mathcal{F}_{t^-}$ represents the information
available before time $t$. The conditional intensity (\ref{condInt})
can be interpreted as the instantaneous rate of point arrivals at time
$t$, given the filtration $\mathcal{F}_{t^-}$.  For a reference
Poisson process, the intensity is independent of the history: it is
constant for a homogeneous process and varies over time in the
inhomogeneous case.

\subsection{Univariate Hawkes process}
The (homogeneous) Hawkes process is an example of a counting process
with a random conditional intensity function:
\begin{equation}
    \label{condIntHawkes}
    \lambda^{\star}(t) = \mu + \sum_{t_i < t} \phi(t - t_i),
\end{equation}
where $\mu >0$ is the baseline intensity, representing the rate of
spontaneous events, and $\phi : \mathbb{R}_+ \longrightarrow
\mathbb{R}_+$ is the excitation function, which quantifies the impact
of past events occurring at times $t_i < t$ on the intensity at the
current time $t$. The dependence on past events induces randomness in
the intensity (\ref{condInt}). A common choice for the excitation
function, adopted in this paper, is the exponential kernel:
\begin{equation}
    \phi(t) = \alpha e^{-\beta t},
\end{equation}
where \( \alpha >0\) determines the value of intensity jump when a new
event occurs, and \( \beta >0 \) controls the rate of its decay.  It
can be shown that the model is stationary for \( 0 < \alpha < \beta
\).  An alternative parametrisation of the model will be used in the
sequel:
\[
    (\rho,\beta)=\left(\frac{\alpha}{\beta},\beta\right),
\]
with \( \beta > 0 \) and \( 0 < \rho < 1 \).  Under this
parametrisation, the excitation function is

\begin{equation}
    \label{expKernel}
    \phi(t)=\rho\beta e^{-\beta t},
\end{equation}
so that \(\int_0^\infty \phi(t)\,dt=\rho\).

According to the immigration-birth representation of the Hawkes
process \citep{Hawkes2}, the parameter $\rho$ can be interpreted as
the expected number of direct offspring generated by a single event.

The exponential kernel is not essential for the modelling framework,
but it is particularly convenient for inference.  More generally, one
may consider any kernel $\phi:\mathbb{R}_+\to[0,\infty)$ such that
  $\phi\in L^1_{\mathrm{loc}}(\mathbb{R}_+)$ and
\[
    \|\phi\|_1:=\int_0^\infty \phi(u)\,du < 1.
\]
The quantity $\|\phi\|_1$ is the branching ratio of the linear Hawkes
process.  The condition $\|\phi\|_1<1$ is the general subcriticality
condition preventing unstable growth of the excitation component.
Likelihood-based inference further requires that the primitive
$\Phi(t)=\int_0^t\phi(u)\,du$ be available in closed form or be
computationally inexpensive.  Likelihood evaluation is particularly
efficient for the exponential kernel or finite mixtures of exponential
kernels thanks to a recursive relation (see Appendix \ref{likcomp}).

\subsection{Covariate-driven extension}
Although the assumption of a constant baseline intensity in the model
(\ref{condIntHawkes}) can be relaxed, models with time-varying
intensity have received comparatively less attention in the
literature. In practice, events corresponding to this part of the
model are rarely spontaneous. In this paper, we assume that they have
been triggered by some exogenous covariates that may vary over
time. Let $X(t)$ denote the covariate process and $x(t)$ a realised
covariate vector, where $X(t)=(X_1(t),\ldots,X_p(t))^\top$ denotes the
$p$-dimensional vector of covariates observed at time $t$. Given
$X(t)=x(t)$, the \textit{covariate-driven Hawkes model} is defined as:

\begin{equation}
    \label{condIntHawkesInh}
    \lambda^{\star}_{\theta}(t) =g\left(\mu+x(t)^\top\gamma\right)  + \sum_{t_i < t} \phi(t - t_i),
\end{equation}
where $\gamma \in \mathbb{R}^p$ is a $p$-dimensional vector of
\textit{regression coefficients}, $\mu$ is a model intercept, and $g :
\mathbb{R} \rightarrow (0, \infty)$ is a suitable link function,
ensuring the positivity of baseline intensity.  In the following, $g$
is exponential, and the excitation function is given by
(\ref{expKernel}). The conditional intensity
($\ref{condIntHawkesInh}$) is indexed by the $(p+3)$-dimensional
vector of model parameters $\theta=(\mu,\gamma,\rho,\beta)$. For prior
elicitation, it will be convenient to work with the equivalent decay
time-scale $\eta=\beta^{-1}$, because prior information is more
naturally expressed on that scale, as explained in the prior section
below.


\section{Bayesian inference for variable selection}
\label{sec:3}
This section specifies likelihood, prior and posterior distributions
of model parameters.

\subsection{Likelihood function}
Let $\{t_1,\ldots,t_{n}\}$ be a point pattern observed on $(0,T]$ for
some fixed observation period $T>0$. Let us assume that the underlying
point process is specified by the conditional intensity
$\lambda^{\star}_\theta(t)$ for
$\theta\in\Theta\subset\mathbb{R}^{p+3}$, given in
\eqref{condIntHawkesInh}. According to \citep{daley}, Proposition
7.3.III, p. 251, the likelihood function is given by the equation:
\begin{equation}
    \label{lik}
    L(t_1,\ldots,t_{n}\mid\theta)=\prod_{i=1}^{n}\lambda^{\star}_{\theta}(t_i)\exp\left(-\int_0^T\lambda_{\theta}^{\star}(u)du\right
    ).
\end{equation}

We restrict attention to covariate processes represented as piecewise
constant on a fixed time grid.  This assumption is mainly
computational: it gives a closed-form expression for the compensator
in the likelihood, since the integral of the covariate-dependent
baseline intensity is then evaluated as a finite sum over the chosen
grid. Such a representation is common and often natural in
applications where covariates are observed, updated, or controlled at
discrete times.  This includes, for example, treatment indicators or
dosage levels that remain unchanged until the next clinical decision,
operational states of a system that change only when a new regime is
imposed, and sensor-based covariates sampled sufficiently often
relative to their intrinsic time scale of variation. More generally,
when the measurement time scale is shorter than the time scale on
which the covariates change appreciably, the values recorded at two
consecutive measurement times provide a natural step-function
approximation of the covariate trajectory. The motivating application
of Section~\ref{scab} falls squarely within this setting, since the
meteorological variables are recorded five times per hour and are not
expected to change appreciably between consecutive measurements.

Let $0=\tau_0 < \tau_1 < \ldots < \tau_{K-1} < \tau_K=T$ be a
partition of the interval $(0,T]$ into a set of subintervals
  $(\tau_{k-1},\tau_{k}]$ of a constant length $c$ for $k=1,\ldots,K$.
    The vector of covariates is piecewise constant, e.g. $\forall t\in
    (\tau_{k-1},\tau_{k}]$, $x(t)=x(\tau_k)=x_{\tau_k}$ for
      $k=1,\ldots,K$. Let us denote by $x_{t_1},\ldots,x_{t_n}$ the
      values of the covariates calculated for the observed time points
      $t_1,\ldots,t_n$.  The log-likelihood function for the
      covariate-driven Hawkes model (\ref{condIntHawkesInh}) is given
      by:

\begin{eqnarray}
    \label{likHawkes}
    l(t_1,\ldots,t_{n}|\theta)&=&\mu+x_{t_1}^{\top}\gamma+\sum_{i=2}^n\log\left( \exp(\mu+x_{t_i}^{\top}\gamma)+\rho\beta A(i)\right)\\
    &-&c\sum_{k=1}^K\exp(\mu+x_{\tau_k}^{\top}\gamma) + \rho\sum_{i=1}^{n}\left(\exp(-\beta(T-t_i))-1\right)\nonumber
\end{eqnarray}
where $A(i)=\sum_{j=1}^{i-1} \exp(-\beta(t_i-t_j))$ for $i>1$ and
$A(1)=0$. The following recurrence relation for $A(i)$ holds:
\begin{equation}
    \label{rec}
    A(i)=\exp\left(-\beta(t_i-t_{i-1}) \right)(1+A(i-1))\text{ for } i\geq 2.
\end{equation}
The derivation of (\ref{likHawkes}) is given in the Appendix
(\ref{likcomp}).  The relationship (\ref{rec}), derived in
\citep{ozaki} for the homogeneous Hawkes processes, brings the cost of
computing the likelihood from quadratic to linear in the number of
observed events $n$.

\subsection{Prior distributions}

Let $Z=(Z_1,\ldots,Z_p)$ denote the latent inclusion indicators
associated with the spike-and-slab prior, and let $\omega$ denote the
common prior inclusion probability.  Although the likelihood is
written in terms of $\beta$, it is more convenient to specify prior
information on the equivalent time-scale $\eta=\beta^{-1}$. We
therefore work in the transformed coordinates $(\mu,\gamma,\rho,\eta)$
for prior elicitation and assume prior independence between the
intercept $\mu$, the branching ratio $\rho$, the kernel time-scale
$\eta$, and the variable-selection layer $(Z,\omega)$, so that
\[
    \pi(\mu,\gamma,\rho,\eta,Z,\omega)
    =
    \pi(\mu)\,\pi(\rho)\,\pi(\eta)\,\pi(\omega)
    \prod_{j=1}^p \pi(Z_j\mid \omega)\,\pi(\gamma_j\mid Z_j).
\]
This induces a corresponding prior $\pi(\theta,Z,\omega)$ in the
original parametrisation $\theta=(\mu,\gamma,\rho,\beta)$.  When
possible, we derive prior distributions with support on the full
admissible parameter space, but we choose to allocate most of the
prior mass to the substantively plausible regions in that space.  This
weakly informative strategy provides regularisation in finite samples
without excluding parameter values that may be strongly supported by
the data, thereby allowing the likelihood to dominate as information
accumulates.

\paragraph{Spike-and-slab prior for $\gamma$.}
The baseline intensity in model (\ref{condIntHawkesInh}) depends on
the linear predictor
\[
    \mu + x(t)^\top \gamma,
\]
where only a subset of the covariates may have a practically
meaningful effect.  To induce sparsity, we place independent
spike-and-slab priors on the regression coefficients
\citep{mitchell,ishwaran}.  For each $j=1,\ldots,p$, let
\[
    Z_j=
    \begin{cases}
        1 & \text{if covariate } x_j \text{ is included}, \\
        0 & \text{otherwise}.
    \end{cases}
\]
Conditionally on $Z_j$, we assume
\[
    \gamma_j \mid Z_j,\sigma_{\text{spike}},\sigma_{\text{slab}}
    \sim
    (1-Z_j)\,\mathcal N(0,\sigma_{\text{spike}}^2)
    +
    Z_j\,\mathcal N(0,\sigma_{\text{slab}}^2),
    \qquad \sigma_{\text{spike}} \ll \sigma_{\text{slab}},
\]
where $\sigma_{\text{spike}} \ll \sigma_{\text{slab}}$ indicates that the spike standard deviation is chosen much smaller than the slab standard deviation.

This is an example of so-called absolutely continuous spike-and-slab
prior \citep{georgeVariableSelectionGibbs1993, ishwaran,
  malsiner2024}.  This choice is standard in the spike-and-slab
literature, where both Dirac spikes and narrow continuous spikes are
commonly used.  The indicators are taken conditionally independent
given the common inclusion probability $\omega$,
\[
    Z_j \mid \omega \sim \mathrm{Bernoulli}(\omega),
    \qquad
    \omega \sim \mathrm{Beta}(a_\omega,b_\omega).
\]

A widely used default is the Scott--Berger choice
$\omega\sim\mathrm{Beta}(1,1)$, which induces a Beta--Binomial prior
on the model size and provides a Bayesian multiplicity correction
relative to the discrete-uniform prior on the model space
\citep{scottberger2010}. Since the present application involves a
small and fixed number of candidate covariates, we adopt this neutral
default rather than a more aggressively sparse prior.

Throughout, the covariates are standardized prior to analysis.  Hence,
for a one-standard-deviation change in $x_j$, the quantity
$\exp(\gamma_j)$ is the multiplicative modification of the baseline
rate.  This gives a direct interpretation to the spike-and-slab
scales.

The spike scale $\sigma_{\text{spike}}$ is chosen to encode
\emph{practical irrelevance}.  More precisely, we regard a covariate
as practically irrelevant if a one-standard-deviation change in that
covariate modifies the baseline rate by at most about $10\%$.  Since
the spike prior is symmetric on the log scale, we choose
\[
    \sigma_{\text{spike}}=\frac{\log(1.1)}{1.96}\approx 0.05.
\]
With this choice, the central $95\%$ prior interval under the spike
component is
\[
    \exp(\gamma_j)\in
    \bigl[\exp(-1.96\sigma_{\text{spike}}),\,\exp(1.96\sigma_{\text{spike}})\bigr]
    =
    [0.91,1.10],
\]
which corresponds to practically negligible modifications of the
baseline rate.

By contrast, the slab scale $\sigma_{\text{slab}}$ is chosen to be
weakly informative on the standardized covariate scale, allowing
substantial but still plausible multiplicative effects on the baseline
rate \citep{gelman2008,joshi2022}.  In the application, we use
$\sigma_{\text{slab}}=1$, which implies the approximate $95\%$ prior
interval
\[
    \exp(\gamma_j)\in
    \bigl[\exp(-1.96\sigma_{\text{slab}}),\,\exp(1.96\sigma_{\text{slab}})\bigr]
    =
    [0.14,7.10].
\]
This interval is deliberately broad.

We use a continuous spike rather than a point mass at $0$, both for
computational convenience and because it permits a direct
interpretation in terms of practical rather than exact irrelevance.  A
comparison between continuous and point-mass spikes is given in
\citep{malsiner2024} and suggests comparable selection performance
between the two.  We also considered horseshoe priors and obtained
similar empirical behaviour, and therefore do not report these
additional results here.

\paragraph{Prior for the intercept $\mu$.}
In a simple Poisson model with constant intensity, a standard Bayesian
choice is a Gamma prior on the intensity scale because it is conjugate
qto the Poisson likelihood \citep[Section 2.6]{gelman2014bda}.  In the
present model, however, the baseline intensity is linked to the
covariates through
\[
    \lambda_0(t)=\exp\bigl(\mu+x(t)^\top\gamma\bigr),
\]
and there is no conjugacy requirement.  It is therefore more natural
to specify the prior directly on the intercept $\mu$, which represents
the \emph{log baseline intensity at average covariate values}.

Weakly informative prior recommendations for generalized linear models
suggest calibrating priors on an interpretable outcome scale rather
than adopting extremely diffuse priors on the coefficient scale
\citep{gelman2008,joshi2022}.  In the present setting, the natural
quantity is the expected number of independent events
(\emph{immigrant} events) over the observation window.  Since the
observation horizon $T$ is known in every application, it is
reasonable to elicit prior information on that scale.  Under the
baseline component alone, the expected number of immigrant events over
$(0,T]$ is
\[
    M_0(\mu)=T\exp(\mu).
\]
We therefore choose a practical lower and upper bound $(m_L,m_U)$ for
the number of immigrant events regarded as plausible a priori, and
define
\[
    \mu \sim \mathcal N(m_\mu,\sigma_\mu^2),
    \qquad
    m_\mu=\frac{\log m_L+\log m_U}{2}-\log T,
    \qquad
    \sigma_\mu=\frac{\log m_U-\log m_L}{2\times 1.96}.
\]
With this choice,
\[
    \mathbb P\!\left(m_L \leq T\exp(\mu)\leq m_U\right)\approx 0.95.
\]

This calibration is intentionally conservative for a Hawkes process.
Indeed, under the branching representation, observed events consist of
immigrant events together with their offspring, and the parameter
$\rho$ represents the expected number of direct offspring generated by
a single event.  Hence, when $\rho>0$, a given total number of
observed events can be achieved with fewer immigrant events than in
the Poisson submodel $\rho=0$.  Choosing the prior on $\mu$ from
plausible total event counts under the worst-case scenario $\rho=0$
therefore ensures that the model remains sufficiently flexible even
before self-excitation is taken into account
\citep{rasmussenBayesianInferenceHawkes2013,serafini2023}.  In the application of
Section~\ref{scab}, the observation window has length $T=1416$, and we
take $(m_L,m_U)=(1,3000)$, which yields
\[
    \mu \sim \mathcal N(-3.25,\,2.04^2).
\]

\paragraph{Prior for the branching ratio $\rho$.}
The branching ratio satisfies $\rho\in(0,1)$ by construction, and we
assign the weakly informative prior
\[
    \rho \sim \mathrm{Beta}(a_\rho,b_\rho).
\]
In the application, we use the uniform choice
$\mathrm{Beta}(1,1)$ on the admissible interval.

\paragraph{Prior for the kernel time-scale $\eta=\beta^{-1}$.}
For the exponential kernel, $\eta$ is the characteristic decay
time-scale of the self-exciting effect.  More precisely, the residual
contribution of a single event at lag $t$ is proportional to
$\exp(-t/\eta)$.  We assign
\[
    \eta \sim \mathcal N^+(0,\sigma_\eta),
\]
where $\mathcal N^+$ denotes the half-normal distribution.  The scale
$\sigma_\eta$ should be chosen from expert knowledge on the duration
of the self-exciting effect.  If $h^\star$ denotes a time horizon
beyond which the residual excitation is regarded as negligible, and
$q$ denotes the corresponding residual fraction, then the condition
\[
    \exp(-h^\star/\eta)\le q
\]
is equivalent to
\[
    \eta \le \frac{h^\star}{\log(1/q)}.
\]
We choose $\sigma_\eta$ such that this inequality holds with prior
probability $0.95$.  In the application, expert knowledge suggests
that less than $5\%$ of the excitation should remain after $24$ time
units, that is, $(h^\star,q)=(24,0.05)$.  This gives the upper bound
\[
    \eta^\star=\frac{24}{\log(20)}\approx 8.01,
\]
and hence
\[
    \mathbb P(\eta\le \eta^\star)=0.95
    \qquad\Longrightarrow\qquad
    \sigma_\eta \approx \frac{8.01}{1.96}=4.09.
\]

Alternatively, one may rescale time by the observation horizon $T$ and
work with the dimensionless time-scale $\tilde\eta=\eta/T$.  If, on
the rescaled time axis, one requires that less than a fraction $q$ of
the excitation remains after a lag $h^\star\in(0,1)$, then
\[
    \exp(-h^\star/\tilde\eta)\le q
    \qquad\Longleftrightarrow\qquad
    \tilde\eta\le \frac{h^\star}{\log(1/q)}.
\]
For example, taking $h^\star=0.1$ and $q=0.05$ yields the upper
bound
\[
    \tilde\eta^\star=\frac{0.1}{\log(20)}\approx 0.033,
\]
and the corresponding half-normal scale is
\[
    \sigma_{\tilde\eta}
    =
    \frac{\tilde\eta^\star}{1.96}
    \approx 0.017
\]
if one wishes this bound to hold with prior probability $0.95$.

\begin{table}[htbp]
    \centering
    \caption{Prior distributions of model parameters and values of
        hyperparameters used in Section~\ref{scab}.}
    \label{tab:priors}

    \renewcommand{\arraystretch}{1.2}
    \rowcolors{2}{gray!15}{gray!30}

    \begin{tabular}{llll}

        \rowcolor{gray!60}
        \color{white} Parameter                                                                  &
        \color{white} Prior distribution                                                         &
        \color{white} Hyperparameters                                                            &
        \color{white} Application                                                                  \\

        $\gamma_j$                                                                               &
        $(1-Z_j)\mathcal{N}(0,\sigma_{\text{spike}}^2)+Z_j\mathcal{N}(0,\sigma_{\text{slab}}^2)$ &
        $\sigma_{\text{spike}},\sigma_{\text{slab}}$                                             &
        $\sigma_{\text{spike}}=0.05,\ \sigma_{\text{slab}}=1$                                      \\

        $Z_j \mid \omega$                                                                        &
        $\mathrm{Bernoulli}(\omega)$                                                             &
        $\omega\sim \mathrm{Beta}(a_\omega,b_\omega)$                                            &
        $a_\omega=1,\ b_\omega=1$                                                                  \\

        $\mu$                                                                                    &
        $\mathcal{N}(m_\mu,\sigma_\mu^2)$                                                        &
        $m_\mu,\sigma_\mu$                                                                       &
        $m_\mu=-3.25,\ \sigma_\mu=2.04$                                                            \\

        $\rho$                                                                                   &
        $\mathrm{Beta}(a_\rho,b_\rho)$                                                           &
        $a_\rho,b_\rho$                                                                          &
        $a_\rho=1,\ b_\rho=1$                                                                      \\

        $\eta=\beta^{-1}$                                                                        &
        $\mathcal{N}^{+}(0,\sigma_\eta)$                                                         &
        $\sigma_\eta$                                                                            &
        $\sigma_\eta=4.09$                                                                         \\
    \end{tabular}
\end{table}

\subsection{Sampling from the posterior distribution}

The likelihood specification and prior distributions lead to the joint
posterior distribution:

\begin{equation*}
    \pi(\theta,Z,\omega\mid t_1,\ldots,t_{n})\propto L(t_1,\ldots,t_{n}\mid\theta) \pi(\theta,Z,\omega).
\end{equation*}

The posterior was marginalised over the discrete inclusion indicators,
and posterior sampling for the continuous parameters was performed
from the resulting marginal posterior using the No-U-Turn Sampler
(NUTS), an adaptive variant of Hamiltonian Monte Carlo
\citep{hoffman}, implemented in RStan \citep{rstan}.  Convergence of
the MCMC chains was assessed using standard RStan diagnostics,
including traceplots, the Gelman–Rubin statistic and effective sample
size calculations.

Posterior inclusion probabilities, denoted by
\begin{equation}
    \label{eq:pip}
    \delta_j = \mathbb P(Z_j=1\mid t_1,\ldots,t_n), \qquad j=1,\ldots,p,
\end{equation}
were then obtained from the conditional distribution of the latent
indicators $Z$ given the sampled continuous parameters. The decision
of including a covariate $x_j$ into the model was based on
thresholding the posterior inclusion probability $\delta_j$ at $0.5$.


\section{Simulation example}

We consider a numerical example to illustrate the spike-and-slab
approach and to assess the effect of increasing the observation
horizon $T$ on posterior concentration. The model is a simplified
version of the one fitted to the real data in Section \ref{scab}. The
prior specification remains the same as in the application (see Table
\ref{tab:priors}). Data were simulated using the algorithm based on
the immigration--birth representation of the Hawkes process
(Appendix~\ref{sec:supp-cluster}).

\paragraph{Observation horizon $T$}

The dataset was generated over the interval $(0,16000]$. To assess the
  effect of an increasing observation horizon on posterior
  concentration, subsets of the initial dataset, restricted to the
  intervals $(0,2000]$, $(0,4000]$, and $(0,8000]$, were considered as
        well.

\paragraph{Baseline intensity $\lambda_{0}(t)$} The baseline intensity involves two periodic
covariates and is specified as:

\begin{equation*}
    \log(\lambda_{0}(t))
    =
    -1.3
    +
    1.2\sin\left(\frac{2\pi t}{24}\right)
    +
    0.05\cos\left(\frac{2\pi t}{24}\right).
\end{equation*}

This choice was motivated by the dataset from Section~\ref{scab},
which includes meteorological covariates with a daily periodic
pattern.

The intercept $\mu=-1.3$ corresponds to a baseline immigration rate of
$\exp(-1.3)\approx 0.273$ events per unit of time, in the absence of
covariate effects. The coefficients $\gamma_1=1.2$ and $\gamma_2=0.05$
jointly determine the amplitude of the daily variation in
$\lambda_0(t)$.  The cosine component is practically irrelevant:
$\gamma_2$ is of the same order of magnitude as the spike standard
deviation ($\sigma_{\text{spike}}=0.05$), while $\exp(\gamma_2)\approx
1.05$ corresponds to about a $5\%$ multiplicative change in the
baseline rate. In contrast, $\gamma_1=1.2$ is of the order of the slab
standard deviation ($\sigma_{\text{slab}}=1$), indicating that the
sine component captures the main daily variation.

\paragraph{Branching ratio $\rho$ and kernel time-scale $\eta$}
The branching ratio $\rho$ was set to $0.3$, corresponding to a
moderate level of self-excitation.  The kernel time-scale parameter
$\eta$ was fixed at $1$, so that, for the exponential kernel, less
than $5\%$ of the excitation remains after the threshold
\begin{equation*}
    h^\star=-\log(0.05)\approx 2.99
\end{equation*}
time units.

\paragraph{Sampling setup} Posterior sampling was performed using the NUTS Sampler
implemented in the RStan library (see Section \ref{sec:3}). Posterior
inference was based on two MCMC chain runs, each consisting of 40,000
iterations, including a warm-up period of 20,000 iterations.

Figure~\ref{fig:sc1:1} displays the prior and posterior distributions
of the coefficients $\gamma_1$ and $\gamma_2$. The posterior
distribution of $\gamma_{1}$ is clearly shifted away from zero,
whereas the posterior distribution of $\gamma_{2}$ is concentrated
around zero. This illustrates the capacity of spike-and-slab to
distinguish between the relevant covariate and the negligible one.

Figure~\ref{fig:sc1:2} displays the evolution of the marginal
posterior distributions of $\mu$, $\gamma_1$, $\gamma_2$, $\rho$ and
$\eta$ as a function of $T$. As the observation interval increases,
the posterior distributions become more concentrated around the true
parameter values. The marginal posterior distributions of the
inclusion probabilities $\delta_1$ and $\delta_2$ (defined in
\eqref{eq:pip}) are also shown in Figure~\ref{fig:sc1:2}. The
posterior distribution of the inclusion probability associated with
the sine covariate is strongly concentrated near 1, while that
associated with the cosine covariate is concentrated near 0.

\begin{figure}
    \centering
    \includegraphics[width=.9\linewidth]{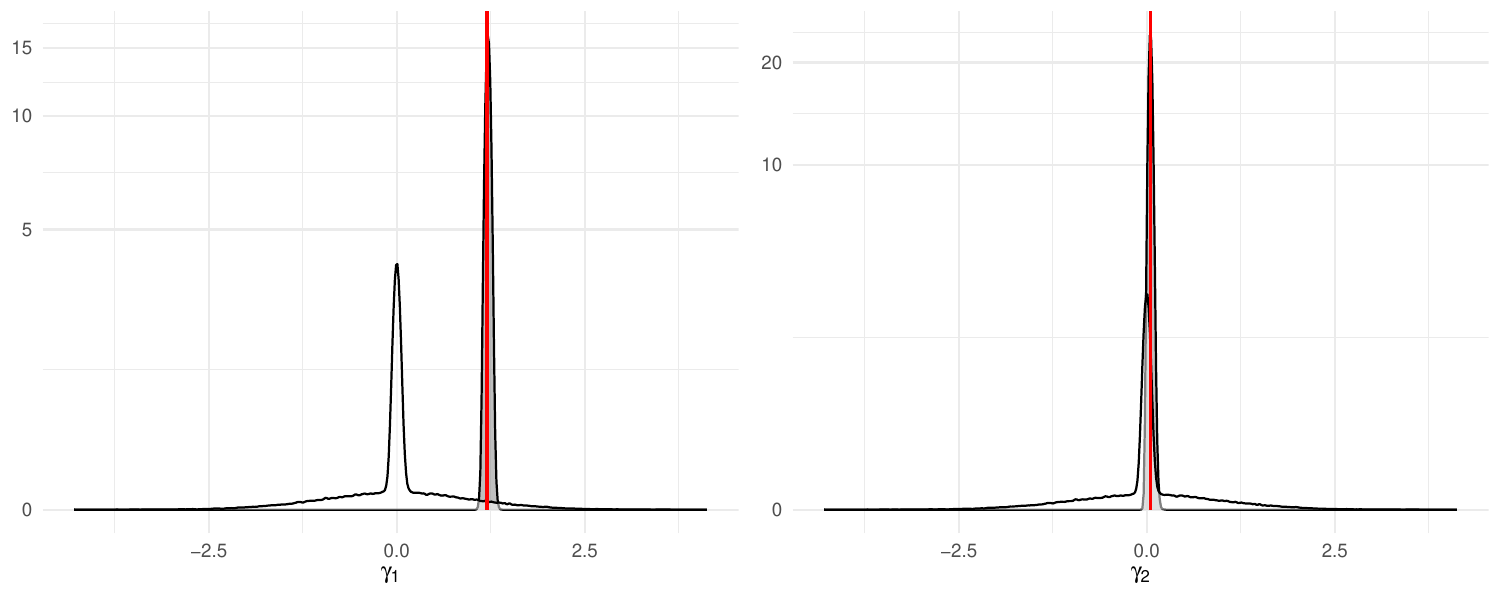}
    \caption{Simulation example: marginal posterior densities of the
        coefficients $\gamma_1=1.2$ and $\gamma_2=0.05$ (grey-shaded
        areas), together with the spike-and-slab prior distributions
        (unshaded areas).}
    \label{fig:sc1:1}
\end{figure}

\begin{figure}
    \centering
    \includegraphics[width=0.9\linewidth]{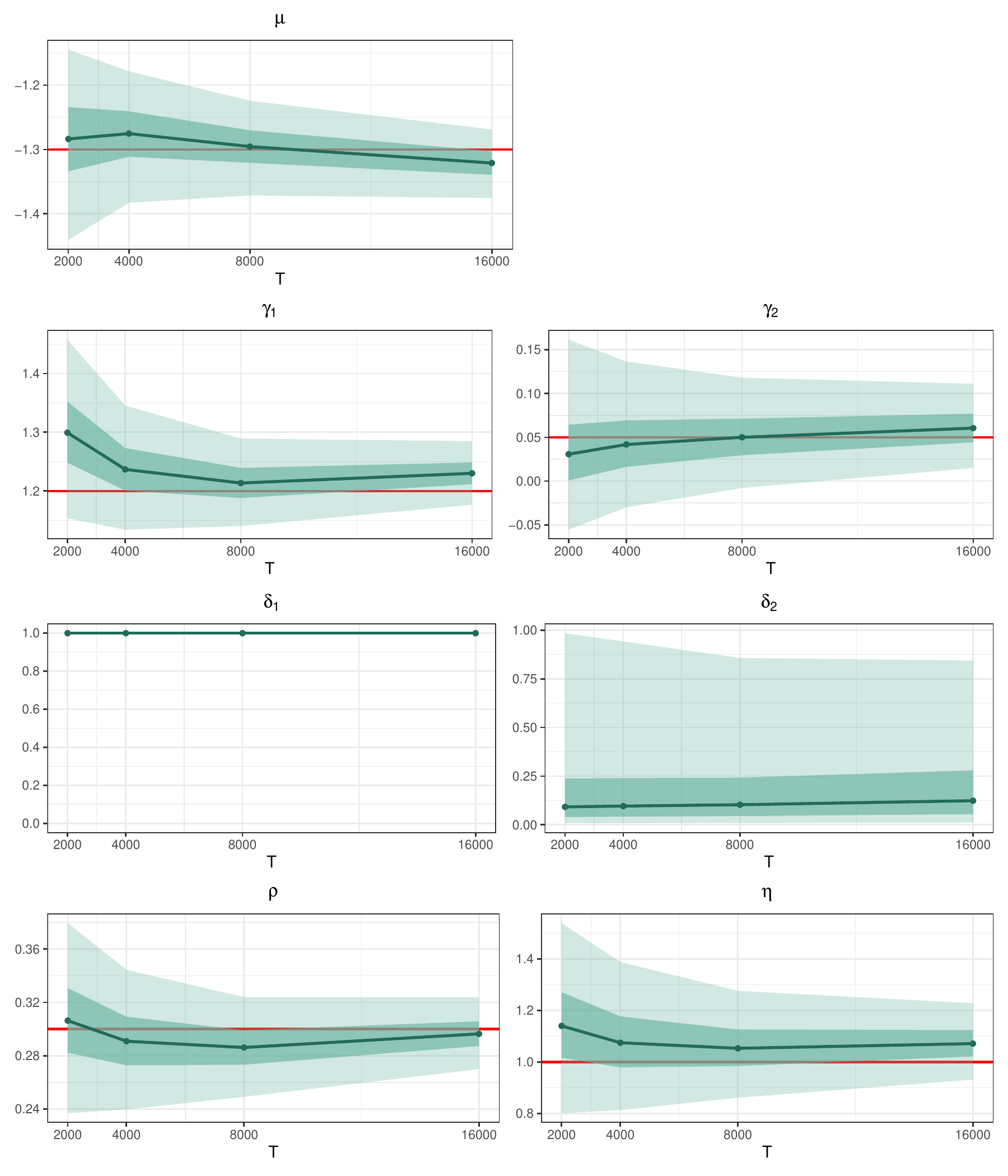}
    \caption{
        Evolution of the posterior distributions as a function of
        the observation horizon T. The dark shaded band represents the
        central 50$\%$ credible interval, while the
        light shaded band represents the 95$\%$ credible
        interval. The solid green line
        indicates the posterior median, and  the horizontal red line indicates
        the true parameter value.}
    \label{fig:sc1:2}
\end{figure}

\section{Application: apple scab risk assessment}
\label{scab}

Apple scab caused by the fungus \textit{Venturia inaequalis} (Cooke)
G. Winter is one of the most serious diseases of apple trees in apple
production areas. It occurs annually on leaves and fruits, causing
significant yield losses. Currently, the most effective way to protect
apple trees against apple scab is to use broad-spectrum mode
fungicides. The number of treatments per year can reach up to 20,
depending on environmental conditions. Accurate forecasting of fungus
ascospore releases is essential to reduce the number of treatments.

Research on apple scab has been conducted worldwide for more than two
centuries (\citep{creemers2006}; \citep{meszka2015}). The life cycle
of \textit{Venturia inaequalis} includes a saprotrophic phase (feeding
on decaying organic matter) and a parasitic phase.  The saprotrophic
phase occurs once per year and develops on fallen apple leaves
infected during the previous growing season, beginning shortly after
leaf fall.  Once mature, ascospores are discharged under favorable
environmental conditions over an average period of 5–9 weeks in the
United States \citep{sutton2014} and 7–11 weeks in Poland, causing
primary infections in apple trees.  Ascospore discharge increases with
the moisture level of infected leaves on the orchard floor, for
example following rainfall or dew formation
\citep{stensvand1998}. After leaf wetting, water enters the
pseudothecia through the ostiole or pores in the pseudothecial wall
\citep{gadoury1998} and subsequently reaches the asci by osmosis
\citep{aylor1991}. As pressure increases within physiologically mature
asci \citep{gadoury1982}, they swell and rupture, leading to ascospore
discharge.  In northern Italy, where climatic conditions differ
substantially from those in Poland, approximately half as many release
events were recorded, and rainfall was identified as the sole trigger
of ascospore release \citep{rossi2001}. In contrast, airborne
ascospores have frequently been detected during dew events in England
and Norway \citep{stensvand1998}. The primary role of leaf wetness was
also observed in Polish conditions \citep{masny2012}.

These results clearly indicate that the effects of weather variables
vary across geographical regions. Commonly used simulation models of
ascospore releases (e.g. A-scab, \citep{rossi2007ascab}) are based on
deterministic approaches that do not account for this source of
variability. Significant discrepancies between predictions and
observed data have been reported under specific weather conditions,
highlighting the limitations of these models \citep{Masny}. A
stochastic framework offers a promising alternative, allowing for the
quantification of forecast uncertainty. In this context, the
covariate-driven Hawkes process is a suitable candidate for modelling
the series of release events and for selecting the relevant weather
variables.

\begin{figure}
    \centering
    \includegraphics[width=\linewidth]{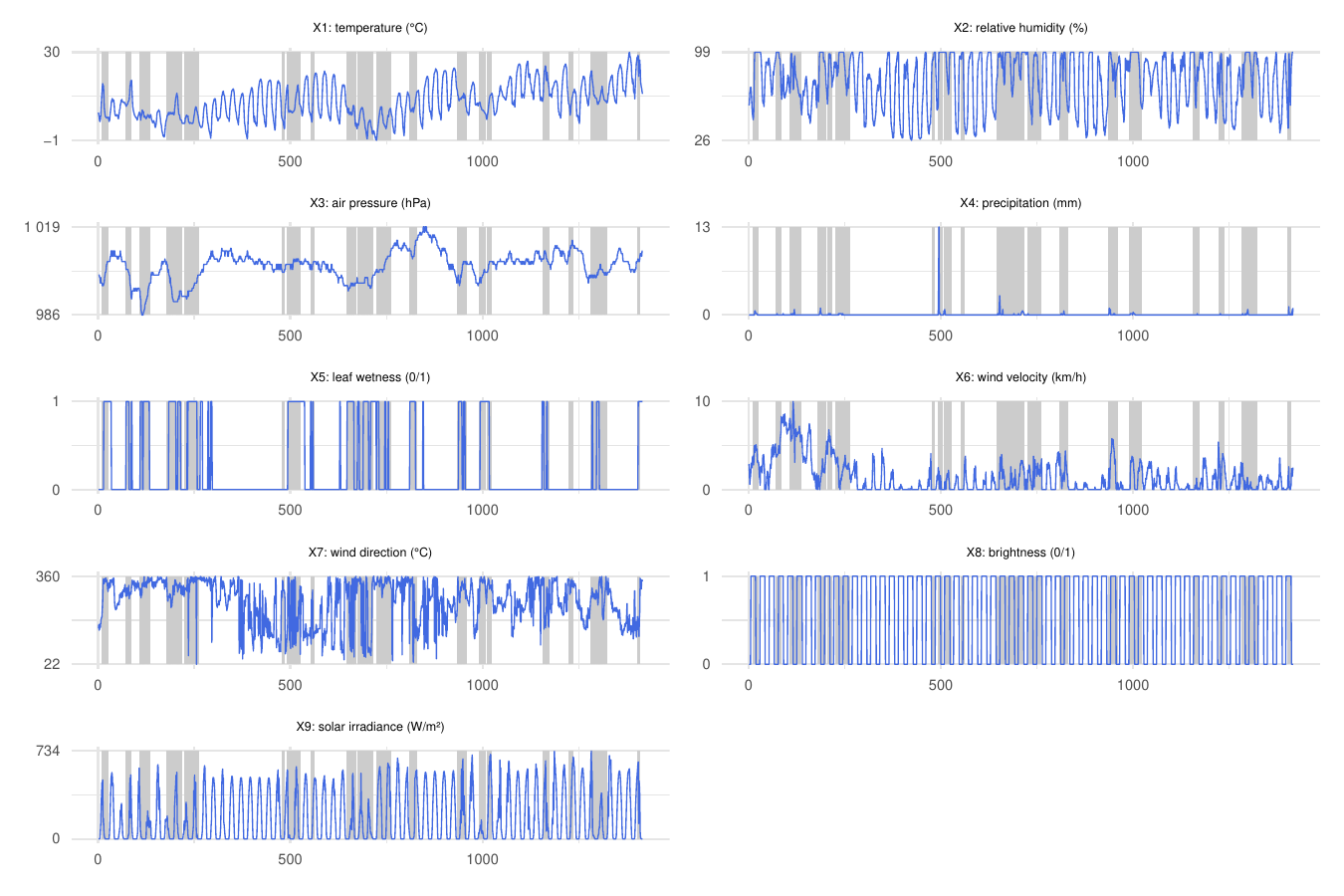}
    \caption{Dynamics of environmental covariates associated with
        spore release in 2011.  Vertical segments represent the instants
        of spore releases. Observation period ranges from $0$ to
        $1416$ hours.}
    \label{fig:covariates}
\end{figure}

\subsection{Experimental data}

Monitoring of ascospore release was carried out in the experimental
orchard of the National Institute of Horticultural Research in
Poland. The study began in early April, when mature ascospores were
abundantly present in the pseudothecia, and ended in the second half
of June. A Burkard spore trap, installed in a plot of an untreated
apple cultivar, recorded spore presence at hourly intervals. The data
used for model fitting were collected in 2011, with 215 spore release
events recorded between April 4 and June 15 (see Figure
\ref{fig:covariates}). Weather data were simultaneously collected by a
local meteorological station at a finer temporal resolution (five
measurements per hour). The hourly means were calculated for
continuous variables, while binary variables were set to their maximum
value across subintervals.  The resulting dataset included nine
environmental covariates: air temperature $X_1(t)$, relative humidity
$X_2(t)$, air pressure $X_3(t)$, precipitation $X_4(t)$, leaf wetness
$X_5(t)$, wind velocity $X_6(t)$, wind direction $X_7(t)$, brightness
$X_8(t)$ and solar irradiance $X_9(t)$ (see Figure
\ref{fig:covariates}). The covariates, displayed here in their
original units, were standardized for analysis.

\subsection{Results}

Posterior samples were obtained from four independent chains, each run
for 50,000 iterations, with the first half of iterations discarded as
warmup to ensure convergence. The convergence diagnostics are provided
in the Appendix \ref{diag}. The priors and the hyperparameters for
this example are listed in Table \ref{tab:priors}.

Figure \ref{fig:gammas} displays the marginal posterior densities of
the regression coefficients together with the spike-and-slab
priors. The densities of $\gamma_1$, $\gamma_4$, $\gamma_6$,
$\gamma_7$, and $\gamma_9$ were shrunk toward zero, whereas the
remaining coefficients had their posterior mass concentrated on values
shifted away from zero and therefore correspond to the selected
covariates. Two of the three moisture-related variables, leaf wetness
and air humidity, have a positive effect on spore releases. In
contrast, the posterior mass of the coefficient $\gamma_4$ associated
with precipitation was shrunk toward the spike component, suggesting
no clear effect.  Higher air pressure can be associated with a
decrease in spore release intensity, with a posterior mean coefficient
$\bar{\gamma}_3=-0.27$.  The results show also that spore release
tends to occur during the daytime ($\bar{\gamma}_8=0.57$).

\begin{figure}[H]
    \centering
    \includegraphics[width=.9\linewidth]{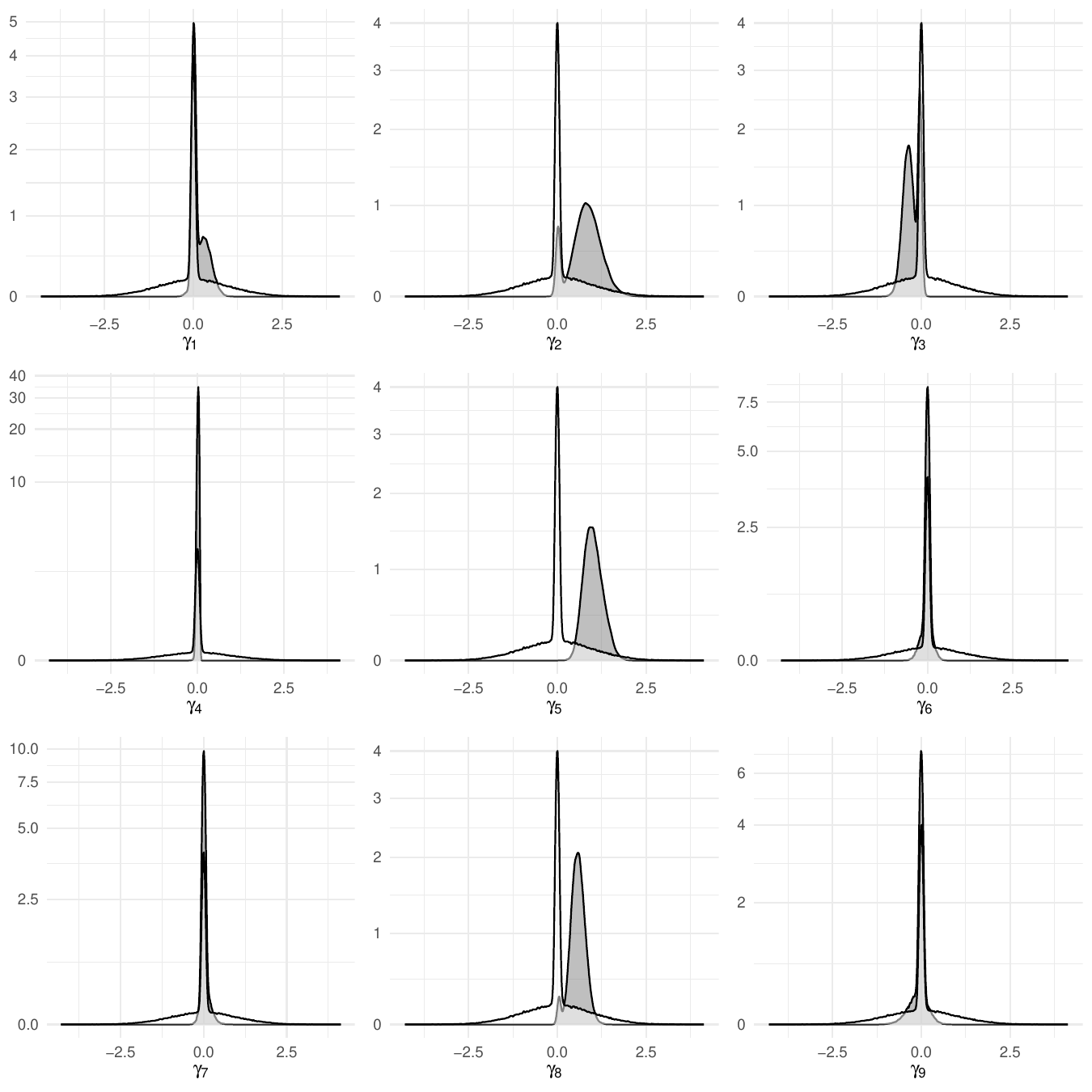}
    \caption{Prior (white) and marginal posterior  densities (grey)
        of the coefficients $\gamma_1,\dots,\gamma_9$. Densities are shown on a
        logarithmic scale to highlight the concentration of
        mass near zero, associated with the spike component.}
    \label{fig:gammas}
\end{figure}

The posterior inclusion probabilities show that leaf wetness and
brightness are included with certainty, relative humidity and air
pressure are included with probabilities $0.89$ and $0.68$
respectively, while the other variables are not selected, with
inclusion probability below the threshold of $0.5$ (see Figure
\ref{fig:pinclusion}).  This illustrates the value of the Bayesian
variable-selection framework: it yields a parsimonious decision rule
while also quantifying uncertainty associated with each inclusion
decision.

\begin{figure}[H]
    \centering
    \includegraphics[width=.6\linewidth]{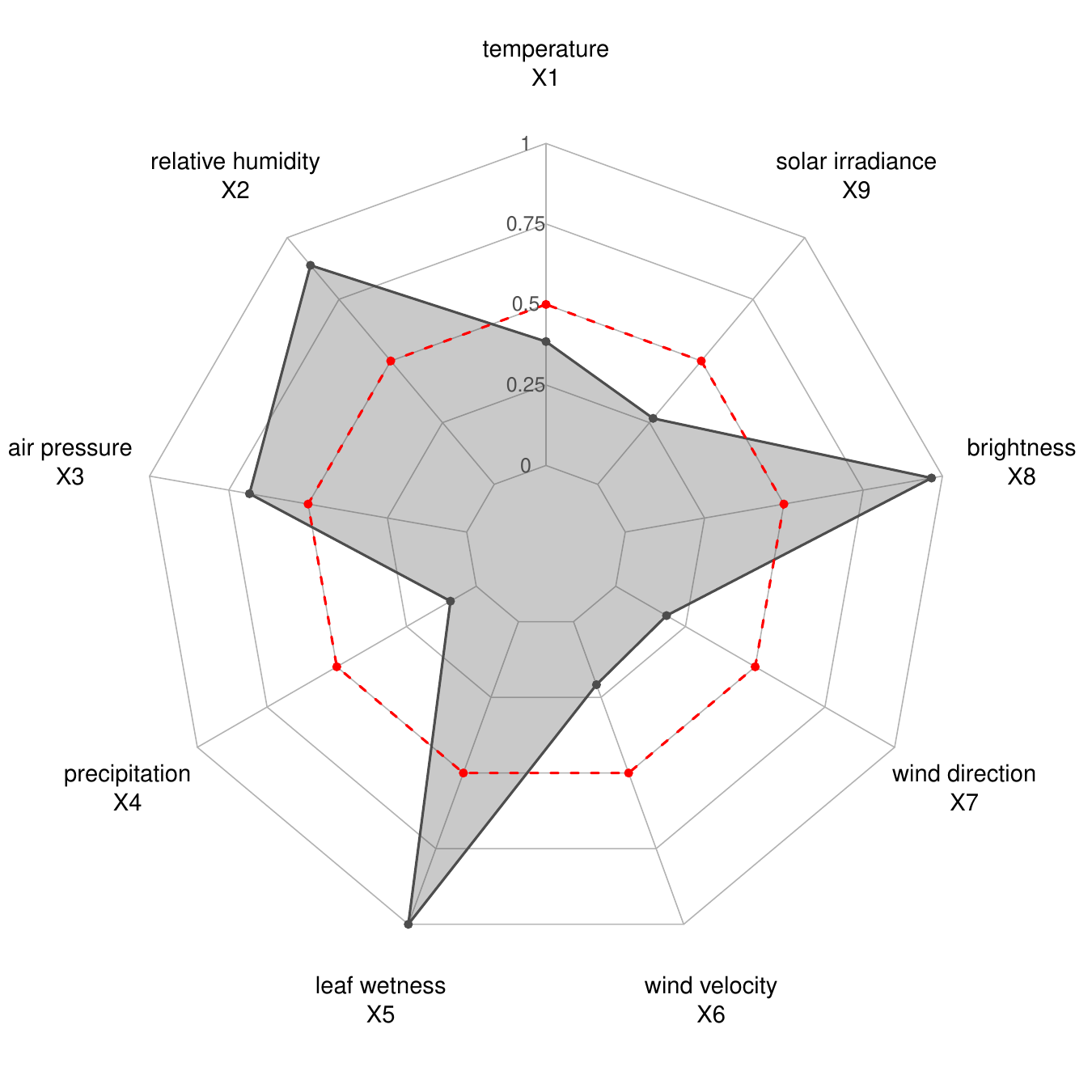}
    \caption{Radar chart of  mean marginal posterior probabilities of inclusion for
        environmental covariates. The red line indicates the inclusion
        threshold of 0.5.}

    \label{fig:pinclusion}
\end{figure}

\begin{table}[htbp]
    \centering
    \caption{Posterior summaries of excitation function parameters with MLE estimates}
    \label{tab:excitation_params}

    \renewcommand{\arraystretch}{1.2}
    \rowcolors{2}{gray!15}{gray!30}

    \begin{tabular}{lllll}

        \rowcolor{gray!60}
        \color{white} Parameter      &
        \color{white} Comments       &
        \color{white} Posterior mean &
        \color{white} 95\% CI        &
        \color{white} MLE estimate     \\

        $\rho$                       &
        branching ratio              &
        0.60                         &
        [0.47, 0.74]                 &
        0.58                           \\

        $\eta$                       &
        kernel time scale            &
        2.87                         &
        [2.12, 3.86]                 &
        2.79                           \\
    \end{tabular}
\end{table}

Table \ref{tab:excitation_params} summarizes marginal posterior
distribution of excitation function parameters, $\rho$ and $\eta$.
The branching ratio $\rho$ quantifies the average number of direct
offspring generated by a single release.  It also provides a
continuous comparison with the covariate-driven inhomogeneous Poisson
submodel: when $\rho=0$, the excitation term vanishes and model
\eqref{condIntHawkesInh} reduces to an inhomogeneous Poisson process
with intensity $\exp\{\mu+x(t)^\top\gamma\}$.  The Hawkes model can
therefore be viewed as a continuous expansion of this Poisson
submodel: rather than comparing two isolated models, the simpler model
is embedded inside a richer family and the parameter $\rho$ measures
the departure from the no-excitation case
\citep[Chapter~7]{gelman2014bda}. The posterior distribution of $\rho$
is concentrated on the interval $[0.47,0.74]$, far from the Poisson
boundary $\rho=0$, providing strong posterior evidence for a
self-exciting component in the spore-release process.  The parameter
$\eta$ quantifies the temporal scale of this effect. Its posterior
distribution is concentrated on the interval $[2.12, 3.86]$.
Consequently, the probability that a single release event triggers
subsequent events (direct offspring) more than 12 hours after its
occurrence remains low (less than $5\%$).

This continuous expansion is also useful for probabilistic
forecasting. Posterior predictive simulations are generated over the
full posterior distribution of $\rho$; hence, uncertainty about the
strength of self-excitation is propagated into the forecast
distribution. If the data supported the no-excitation case, posterior
draws of $\rho$ would concentrate near zero and the corresponding
forecasts would approach those of the covariate-driven Poisson
submodel.

Out-of-sample posterior predictive checks were performed to assess the
goodness-of-fit of the model. The model was trained on the data
recorded over the interval $(0, 2/3T)$, where $T$ denotes the
observation period. The prediction was then carried out over the
remaining interval $(2/3T, T)$.  Sample paths of the process were
generated using parameters drawn from the posterior distribution
fitted to the training data. More precisely, a random subset of $500$
posterior samples was selected, and the Hawkes process was simulated
according to Ogata’s modified thinning algorithm, adapted to the case
of a covariate-dependent baseline (see Appendix~\ref{sec:supp-ogata}).
The 95\% posterior predictive intervals encapsulate the observed data
and the median is in general close to them, except from the last phase
when the median number of events slightly underestimates the actual
number of events (see Figure~\ref{fig:outofsample}).

\begin{figure}
    \centering
    \includegraphics[width=.7\linewidth]{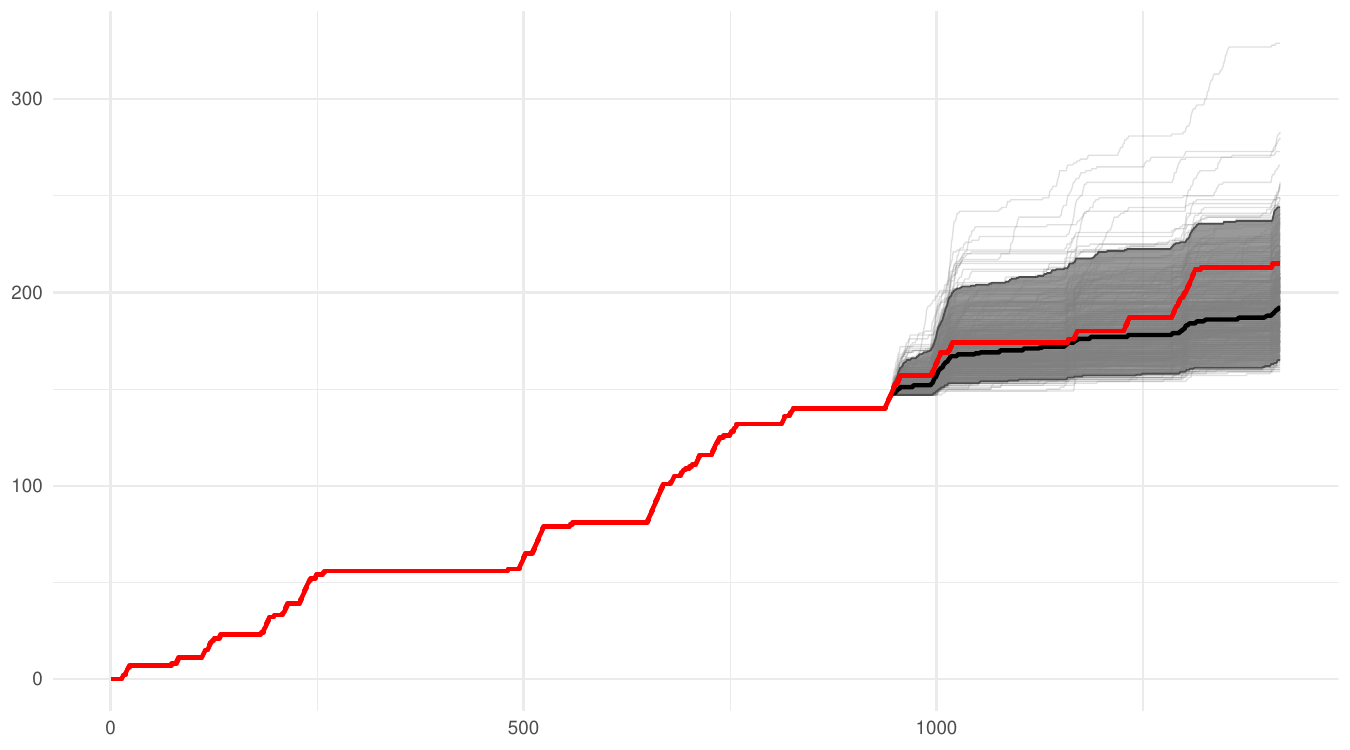}

    \caption{Out-of-sample model validation. Red line: cumulative
        counts of spore releases over the observation period $(0,T)$.
        Grey lines: 500  Hawkes paths, simulated over the interval
        $(2/3T,T)$. Grey area: 95\% posterior predictive bands.  Black
        line: pointwise median of cumulative counts across simulations.}
    \label{fig:outofsample}
\end{figure}

\subsection{Probabilistic forecasting experiment}

We evaluated the accuracy of out-of-sample prediction, considering the
data collected in 2012. The observation period was divided into
subintervals of 48 hours, and three intervals with low, medium and
high intensity of release events were selected. The high-intensity
period is associated with low air pressure and the occurrence of leaf
wetness at the beginning of the interval, whereas the medium-intensity
period is associated with these conditions approximately 15 hours
later (see Figure \ref{fig:forecast}, lower panel).  For each
scenario, the model was trained on the data from 2011 and the subset
of data from $2012$, preceding the forecast period. The simulations of
Hawkes processes with parameters sampled randomly from the posterior
distribution were then performed, and their pointwise summaries were
calculated. Figure \ref{fig:forecast} displays $95\%$ predictive bands
and median calculated from the simulations.  The model captures the
contrast between the three scenarios correctly. The $95\%$ posterior
predictive intervals encapsulate the observed values and the median of
simulated counts turns out to be a reliable indicator of the release
risk.  In this retrospective forecasting experiment, the covariate
trajectories over the forecast windows are treated as known. In an
operational setting, they would be replaced by meteorological
forecasts.

\begin{figure}[H]
    \centering
    \includegraphics[width=\linewidth]{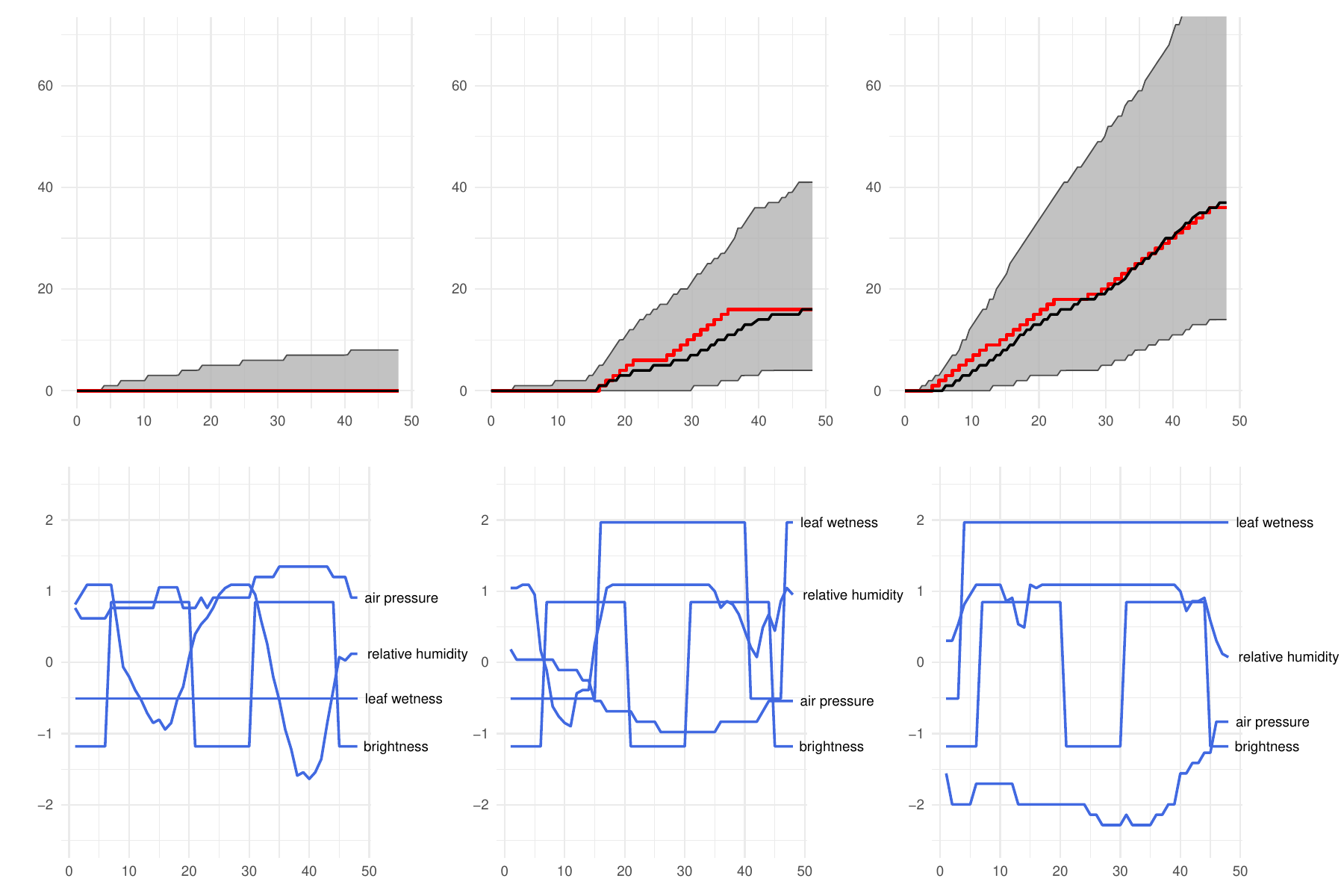}

    \caption{Forecasting experiment. First row: $95\%$ predictive bands
        (grey area) and median (black line) of release cumulative
        counts, calculated from $1000$ simulations of the Hawkes process
        with parameters sampled from the posterior distribution. The
        observed cumulative counts are displayed in red. The panels
        correspond to the scenarios of low (left), intermediate (middle)
        and high (right) risk of spore releases for the observation period
        of 48h. Second row: temporal evolution of air pressure, leaf
        wetness, brightness and relative humidity, variables selected by
        spike-and-slab approach.}
    \label{fig:forecast}
\end{figure}


\section{Conclusions}

This paper focuses on the univariate Hawkes process with a baseline
intensity that depends on a set of dynamic covariates through a
log-linear link function. We adopted a Bayesian approach with
spike-and-slab priors on the regression coefficients to perform
variable selection. This approach enables the identification of
covariates with practically meaningful effects, while retaining
posterior uncertainty through inclusion probabilities. It also
provides a natural framework for probabilistic forecasting, since
uncertainty in the regression coefficients and in the excitation
parameters can be propagated into posterior predictive simulations.

The model was applied to analyse the temporal pattern of
\textit{Venturia inaequalis} spore releases in the context of apple
scab risk assessment. The results highlight the effect of four
meteorological covariates on primary releases and demonstrate the
existence of a short-term self-exciting component. The posterior
distribution of the branching ratio was concentrated away from zero,
supporting the use of the Hawkes expansion over the covariate-driven
inhomogeneous Poisson submodel. The predictive capacity of the model
was evaluated through posterior predictive checks and a probabilistic
forecasting experiment. The model provides an indicator of release
risk, together with an assessment of prediction uncertainty, provided
that observed or forecast meteorological covariates are available.

The framework could be extended in several directions. First, although
the exponential triggering kernel is convenient and leads to efficient
likelihood evaluation through a recursive relation, the modelling
principle is not restricted to this choice. More general triggering
kernels could be considered, provided that the usual subcriticality
condition is satisfied. In that case, however, the likelihood would
generally require either a closed-form primitive of the kernel or
numerical evaluation of the excitation contribution, whereas the
exponential kernel gives a particularly efficient implementation.

Second, the piecewise-constant representation of the covariates is a
computationally convenient choice, because it yields a closed-form
expression for the baseline contribution to the compensator. If one
wished to interpolate covariate trajectories between measurement
times, the same model could be used with numerical integration of the
covariate-dependent baseline intensity. A further extension would be
to model the covariate trajectory itself as stochastic or partially
observed, rather than conditioning on a fixed realised
trajectory. This would connect the present framework to recent work on
Hawkes processes with diffusion-driven baseline intensities
\citep{perrinHawkesProcessDiffusiondriven2025}, and would allow
uncertainty in the covariate process to be propagated into inference
and forecasting.

More broadly, the proposed approach can be applied to other temporal
event data in which exogenous covariates influence the baseline rate
while past events induce short-term clustering. In such settings, the
Bayesian framework provides both interpretable variable selection and
probabilistic forecasts of future event occurrence.

\paragraph{{\bf Acknowledgements}}
This research was supported by DATAIA Convergence Institute as part of
the « Programme d’Investissement d’Avenir », (ANR-17-CONV-0003)
operated by the French National Research Institute for Agriculture, Food
and Environment.


\appendix
\numberwithin{equation}{section}
\renewcommand{\theequation}{\thesection.\arabic{equation}}
\gdef\theHsection{appendix.\arabic{section}}
\gdef\theHsubsection{appendix.\arabic{section}.\arabic{subsection}}
\gdef\theHequation{appendix.\thesection.\arabic{equation}}


\section{Computation of the likelihood function}
\label{likcomp}
The first term of the likelihood \eqref{lik} for the model \eqref{condIntHawkesInh} is given by:
\begin{align}
    \label{intensities}
    \prod_{i=1}^{n}\lambda^{\star}_{\theta}(t_i)
     & =
    \exp(\mu+x_{t_1}^{\top}\gamma)\prod_{i=2}^{n}\left(\exp(\mu+x_{t_i}^{\top}\gamma) +\rho\beta\sum_{j=1}^{i-1} \exp(-\beta(t_i-t_j))\right)\nonumber \\
     & =
    \exp(\mu+x_{t_1}^{\top}\gamma)\prod_{i=2}^{n}\left[\exp(\mu+x_{t_i}^{\top}\gamma) +\rho\beta A(i)\right],
\end{align}
where for $i\geq 2$,
\[
    A(i) = \sum_{j=1}^{i-1} \exp\{-\beta(t_i-t_j)\},
\]
with $A(1)=0$. For efficient $O(n)$ evaluation, $A(i)$ satisfies the recursion
\[
    A(i) = \exp\{-\beta(t_i-t_{i-1})\}\{1+A(i-1)\}, \qquad i=2,\dots,n.
\]
Recalling the partition $0=\tau_0 < \tau_1 < \dots < \tau_K=T$ of
$(0,T]$ into subintervals $(\tau_{k-1},\tau_k]$ of constant length
    $c=\tau_k-\tau_{k-1}$ on which $x(t)=x_{\tau_k}$, the compensator
    integral in \eqref{lik} is evaluated as follows:
\begin{align}
    \label{compensator}
    \int_0^T\lambda_{\theta}^{\star}(t)\,dt
     & =
    \int_0^T\exp(\mu+x(t)^{\top}\gamma)\,dt+\int_0^T\rho\beta\sum_{t_i<t} \exp(-\beta(t-t_i))\,dt\nonumber                                       \\
     & =
    \sum_{k=1}^K\int_{\tau_{k-1}}^{\tau_k}\exp(\mu+x(t)^{\top}\gamma)\,dt+\rho\beta\sum_{i=1}^{n}\int_{t_i}^{T} \exp(-\beta(t-t_i))\,dt\nonumber \\
     & =
    c\sum_{k=1}^K\exp(\mu+x_{\tau_k}^{\top}\gamma)-\rho\sum_{i=1}^{n}\left(\exp(-\beta(T-t_i))-1\right).
\end{align}
Subtracting the compensator \eqref{compensator} from the logarithm of
\eqref{intensities} yields the log-likelihood in \eqref{likHawkes}.


\section{Simulation of the covariate-driven Hawkes process}
\label{sec:supp-simulation}

Let $t_0<T$ denote the beginning and end of the simulation window, and let
\[
    \mathcal H_0=\{t_i:t_i< t_0\}
\]
be the observed event history available at $t_0$.  The history may be
empty, in which case the algorithms provide unconditional forward
simulation from an empty initial condition.  Conditional prediction is
useful for forecasting based on observed past events and covariate
trajectories (as in Section~\ref{scab}).  Conditional on a realised
covariate trajectory $x(t)$ and fixed parameter values, the
conditional intensity on $(t_0,T]$ is
\begin{equation}
    \label{eq:supp-intensity}
    \lambda^\star(t)
    =
    \lambda_0(t)
    +
    E(t),
    \qquad
    \lambda_0(t)=g\{\mu+x(t)^\top\gamma\},
\end{equation}
where $\beta>0$, $0\leq\rho<1$, and
\begin{equation}
    \label{eq:supp-excitation}
    E(t) = \rho\beta\sum_{t_i<t}\exp\{-\beta(t-t_i)\}
\end{equation}
denotes the excitation state (the self-exciting component) at time
$t$.  In particular, at the forecast origin $t_0$, the excitation
generated by the observed history $\mathcal H_0$ is given by
\begin{equation}
    \label{eq:supp-initial-excitation}
    E_0 = E(t_0) = \rho\beta\sum_{t_i\in\mathcal H_0}\exp\{-\beta(t_0-t_i)\}.
\end{equation}
We assume throughout that $\lambda_0$ is
measurable and bounded on $(t_0,T]$, and choose $M>0$ such that
\[
    \sup_{t\in(t_0,T]}\lambda_0(t)\leq M.
\]

The algorithms below are adaptations of two standard constructions for
linear Hawkes processes.  The first applies the modified thinning
method of \citet{ogataLewisSimulationMethod1981}, which extends the
Poisson thinning method of
\citet{lewisSimulationNonhomogeneousPoisson1979} to point processes
specified by predictable conditional intensities.  The second uses the
immigration--birth, or Poisson cluster representation of
\citet{Hawkes2}.  These two constructions are standard in the Hawkes
process literature; see also \citet{daley}, \citet{bernabeu}, and
\citet{laub}.  Our purpose is to state their adaptation to the
covariate-dependent baseline in \eqref{eq:supp-intensity} and to allow
simulation conditional on a possibly empty history.

\subsection{Modified thinning conditional on the observed history}
\label{sec:supp-ogata}

The simulation method of \citet{ogataLewisSimulationMethod1981} extends the
thinning construction of \citet{lewisSimulationNonhomogeneousPoisson1979} to
point processes specified by predictable conditional intensities. The following
theorem states the modified-thinning principle in the univariate form needed here.
It is a direct specialisation of the predictable dominating-process construction
and Proposition~I of \citet{ogataLewisSimulationMethod1981}.

\begin{theorem}[Modified thinning with a predictable dominating intensity]
    \label{thm:supp-modified-thinning}
    Let \(\lambda^\star\) be a non-negative predictable intensity process on
    \((t_0,T]\). Suppose that one can construct a finite, positive, predictable,
    piecewise-constant process \(B=\{B(t):t\in(t_0,T]\}\), possibly depending on
    the previously accepted event history, such that
    \[
        \lambda^\star(t)\leq B(t), \qquad t\in(t_0,T],
    \]
    almost surely. Conditional on the history, generate a point process of
    candidates with predictable intensity \(B(t)\). At each candidate time \(s\),
    conditionally on the information available immediately before \(s\),
    independently draw \(U_s\sim\mathcal U(0,1)\) and retain the candidate if
    \[
        U_s\leq \frac{\lambda^\star(s)}{B(s)},
    \]
    where both \(\lambda^\star(s)\) and \(B(s)\) are evaluated from the predictable
    state immediately before the acceptance decision. Then the retained candidate
    times form a point process on \((t_0,T]\) with conditional intensity
    \(\lambda^\star\).
\end{theorem}

\begin{proof}
    This is the univariate specialisation of the predictable dominating-process
    construction and Proposition~I of \citet{ogataLewisSimulationMethod1981}, with
    mark \(1\) denoting acceptance and mark \(0\) denoting rejection.
\end{proof}

In practice, on any interval over which the dominating process is constant,
say \(B(u)\equiv B\), the next candidate waiting time is exponentially
distributed with rate \(B\).

Algorithm~\ref{alg:supp-ogata} applies this principle to the covariate-driven
Hawkes process with exponential excitation.

\begin{algorithm}[H]
    \linespread{1.15}\selectfont
    \caption{Conditional simulation by modified thinning}
    \label{alg:supp-ogata}
    \KwIn{$t_0<T$, history $\mathcal H_0$, parameters $(\mu,\gamma,\rho,\beta)$,
        covariate trajectory $x(t)$, link function $g$, and $M>0$ satisfying
        $\sup_{t\in(t_0,T]}\lambda_0(t)\leq M$}
    \KwOut{Future events $\mathcal E$ on $(t_0,T]$}

    $\lambda_0(t)\leftarrow g\{\mu+x(t)^\top\gamma\}$\;
    $\mathcal E\leftarrow\varnothing$, $t\leftarrow t_0$\;
    $E\leftarrow\rho\beta\sum_{t_i\in\mathcal H_0}
        \exp\{-\beta(t_0-t_i)\}$ \tcp*[r]{initial excitation state $E_0=E(t_0)$}\;

    \While{$t<T$}{
    $B\leftarrow M+E$ \tcp*[r]{local dominating rate}
    Draw $u\sim\mathcal U(0,1)$ and set
    $s\leftarrow t-B^{-1}\log u$\;
    \If{$s\geq T$}{
        \textbf{break}\;
    }
    $E_s\leftarrow E\exp\{-\beta(s-t)\}$\;
    $\lambda_s\leftarrow\lambda_0(s)+E_s$\;
    Draw $v\sim\mathcal U(0,1)$\;
    \eIf{$v\leq\lambda_s/B$}{
        $\mathcal E\leftarrow\mathcal E\cup\{s\}$\;
        $E\leftarrow E_s+\rho\beta$ \tcp*[r]{accepted event creates a jump}
    }{
        $E\leftarrow E_s$ \tcp*[r]{rejected proposal creates no jump}
    }
    $t\leftarrow s$\;
    }
    \Return{$\mathcal E$}\;
\end{algorithm}

The correctness of Algorithm~\ref{alg:supp-ogata} relies on tracking the
scalar excitation state and maintaining a valid predictable bound on each
inter-candidate interval.

\begin{lemma}[Excitation recursion and local domination]
    \label{lem:supp-ogata-state}
    At every update time \(t\) of Algorithm~\ref{alg:supp-ogata}, the state variable satisfies
    \[
        E = \rho\beta \sum_{t_i\in\mathcal H_0\cup\mathcal E:\,t_i\leq t} \exp\{-\beta(t-t_i)\}.
    \]
    If \(s>t\) is the next candidate time, then
    \[
        E(s^-)=E\exp\{-\beta(s-t)\},
    \]
    and the assignments in Algorithm~\ref{alg:supp-ogata} preserve this identity after acceptance or rejection.
    Moreover,
    \[
        \lambda^\star(u)\leq M+E, \qquad u\in[t,s].
    \]
\end{lemma}

\begin{proof}
    At initialisation \(t=t_0\), \(\mathcal E=\varnothing\), so \(E\) matches the excitation
    \(E_0=E(t_0)\) generated by \(\mathcal H_0\) (with \(E_0=0\) if \(\mathcal H_0=\varnothing\)).
    For a candidate proposal \(s>t\), factoring the exponential decay as
    \(\exp\{-\beta(s-t_i)\}=\exp\{-\beta(s-t)\}\exp\{-\beta(t-t_i)\}\) yields
    \begin{align*}
        E(s^-)
         & =
        \rho\beta\sum_{t_i\leq t}\exp\{-\beta(s-t_i)\} \\
         & =
        \exp\{-\beta(s-t)\}\left[\rho\beta\sum_{t_i\leq t}\exp\{-\beta(t-t_i)\}\right]
        =
        E\exp\{-\beta(s-t)\}
        =
        E_s.
    \end{align*}
    If \(s\) is rejected, no event occurs and \(E(s)=E(s^-)=E_s\), matching \(E\leftarrow E_s\).
    If \(s\) is accepted, the new event adds an instantaneous jump of size \(\phi(0)=\rho\beta\),
    so \(E(s)=E_s+\rho\beta\), matching \(E\leftarrow E_s+\rho\beta\).
    By induction, the state variable \(E\) tracks the true excitation without approximation.
    Finally, since \(\beta>0\) and \(u\geq t\), \(\exp\{-\beta(u-t)\}\leq 1\), so
    \(E(u)=E\exp\{-\beta(u-t)\}\leq E\) on \([t,s]\). Combined with \(\lambda_0(u)\leq M\),
    we obtain \(\lambda^\star(u)=\lambda_0(u)+E(u)\leq M+E=B\) on \([t,s]\), and the proposal
    is accepted with probability \(\lambda_s/B = \lambda^\star(s)/B\).
\end{proof}

\begin{proposition}[Validity of Algorithm~\ref{alg:supp-ogata}]
    \label{prop:supp-ogata}
    Assume that \(\lambda_0\) is measurable and bounded on \((t_0,T]\), and let
    \(M>0\) satisfy \(\lambda_0(u)\leq M\) for all \(u\in(t_0,T]\). Then,
    conditional on the observed history \(\mathcal H_0\) and covariate trajectory
    \(x\), Algorithm~\ref{alg:supp-ogata} generates a point process on \((t_0,T]\)
    with conditional intensity
    \[
        \lambda^\star(u)
        =
        \lambda_0(u)
        +
        \rho\beta\sum_{t_i<u}\exp\{-\beta(u-t_i)\},
    \]
    and terminates after finitely many proposals almost surely.
\end{proposition}

\begin{proof}
    By Lemma~\ref{lem:supp-ogata-state}, the process \(B\) is predictable, piecewise
    constant, and dominates \(\lambda^\star\) on each inter-candidate interval. At each
    candidate time \(s\), Algorithm~\ref{alg:supp-ogata} accepts with probability
    \[
        \frac{\lambda_s}{B} = \frac{\lambda^\star(s)}{B(s)}.
    \]
    Theorem~\ref{thm:supp-modified-thinning} therefore shows that the accepted points
    \(\mathcal E\) form a point process on \((t_0,T]\) with conditional intensity
    \(\lambda^\star\).

    To establish almost-sure termination, note that by the subcritical
    immigration--birth representation of \citet{Hawkes2}, the total number of events
    on the finite forecast window is almost surely finite when \(0\leq\rho<1\)
    (\(|\mathcal E|<\infty\)). Because excitation increases only at accepted events by
    \(\rho\beta\) and decays otherwise, the dominating rate is bounded along any
    sample path by
    \[
        \sup_{u\in(t_0,T]} B(u) \leq M+E_0+\rho\beta|\mathcal E| < \infty.
    \]
    Consequently, the total number of candidate proposals drawn from the dominating
    Poisson process on \((t_0,T]\) is almost surely finite, guaranteeing that
    Algorithm~\ref{alg:supp-ogata} terminates after finitely many proposals almost surely.
\end{proof}

\subsection{Conditional simulation from the cluster representation}
\label{sec:supp-cluster}

The immigration--birth representation of the linear Hawkes process
\citep{Hawkes2} separates immigrant events, generated by the baseline
intensity \(\lambda_0\), from their descendants.  Conditional on the
immigrant process, each event independently generates a Poisson
process of direct offspring, and every offspring becomes the parent of
a new generation.

For the exponential kernel, the direct offspring of a parent occurring at time \(t_p\) form a Poisson process with intensity
\[
    \phi(t)=\rho\beta e^{-\beta (t-t_p)},\qquad t>t_p.
\]

To simulate only offspring whose occurrence times fall in the forecast window, define
\[
    a(t_p)=\max(0,t_0-t_p), \qquad b(t_p)=T-t_p.
\]

This is because if $t_p$ is before $t_0$, we make a  prediction in between $(t_0, T]$, else in $(t_p, T]$.
The mean number of direct offspring of $t_p$ whose occurrence times lie in
$(t_0,T]$ is, by integration of the offspring intensity over this time interval,
\begin{equation}
    \label{eq:supp-truncated-mean}
    m(t_p)
    =
    \rho\{e^{-\beta a(t_p)}-e^{-\beta b(t_p)}\},
\end{equation}
provided $b(t_p)>a(t_p)$.  Conditional on such an offspring being in the
window, its lag $D$ has the exponential distribution truncated to
$(a(t_p),b(t_p)]$.
It can be generated via the inverse-transform method from $u\sim\mathcal U(0,1)$ by
\begin{equation}
    \label{eq:supp-truncated-lag}
    D
    =
    -\frac{1}{\beta}
    \log\left[
    e^{-\beta a(t_p)}
    -u\{e^{-\beta a(t_p)}-e^{-\beta b(t_p)}\}
    \right].
\end{equation}

\begin{algorithm}[H]
    \linespread{1.15}\selectfont
    \caption{Conditional simulation by the cluster representation}
    \label{alg:supp-cluster}
    \KwIn{$t_0<T$, history $\mathcal H_0$, parameters $(\mu,\gamma,\rho,\beta)$,
        covariate trajectory $x(t)$, link function $g$, and $M>0$ satisfying
        $\sup_{t\in(t_0,T]}\lambda_0(t)\leq M$}
    \KwOut{Future events $\mathcal E$ on $(t_0,T]$}

    $\lambda_0(t)\leftarrow g\{\mu+x(t)^\top\gamma\}$\;

    \textbf{Step 1: simulate future immigrants by thinning}\;
    $t\leftarrow t_0$, $\mathcal I\leftarrow\varnothing$\;
    \While{$t<T$}{
    Draw $u_1\sim\mathcal U(0,1)$ and set
    $t\leftarrow t-M^{-1}\log u_1$\;
    \If{$t\geq T$}{
        \textbf{break}\;
    }
    Draw $u_2\sim\mathcal U(0,1)$\;
    \If{$u_2\leq\lambda_0(t)/M$}{
        $\mathcal I\leftarrow\mathcal I\cup\{t\}$\;
    }
    }

    \textbf{Step 2: generate descendants in the forecast window}\;
    $\mathcal E\leftarrow\mathcal I$\;
    $\mathcal P\leftarrow\operatorname{queue}(\mathcal H_0\cup\mathcal I)$\;
    \While{$\mathcal P$ is not empty}{
    Remove the first parent time $t_p$ from $\mathcal P$\;
    $a\leftarrow\max(0,t_0-t_p)$, $b\leftarrow T-t_p$\;
    \If{$b>a$}{
    $m\leftarrow\rho\{e^{-\beta a}-e^{-\beta b}\}$\;
    Draw $N\sim\operatorname{Poisson}(m)$\;
    \For{$j=1,\ldots,N$}{
    Draw $u\sim\mathcal U(0,1)$\;
    $D\leftarrow-\beta^{-1}\log[e^{-\beta a}
        -u\{e^{-\beta a}-e^{-\beta b}\}]$\;
    $t_e\leftarrow t_p+D$\;
    $\mathcal E\leftarrow\mathcal E\cup\{t_e\}$\;
    Append $t_e$ to $\mathcal P$ \tcp*[r]{offspring becomes a parent}
    }
    }
    }
    \Return{$\operatorname{sort}(\mathcal E)$}\;
\end{algorithm}

\begin{proposition}[Validity of Algorithm~\ref{alg:supp-cluster}]
    \label{prop:supp-cluster}
    Under the assumptions stated above, Algorithm~\ref{alg:supp-cluster} generates,
    conditional on $(\mathcal H_0,x)$, a point process on $(t_0,T]$ with
    conditional intensity \eqref{eq:supp-intensity}.
\end{proposition}

\begin{proof}
    The validity follows from Poisson thinning and the immigration--birth
    representation of linear Hawkes processes, after verifying the conditioning
    and truncation ingredients for the present model.

    \medskip
    \noindent
    \textit{1. Immigrant generation.}
    By the thinning theorem of \citet{lewisSimulationNonhomogeneousPoisson1979},
    Step~1 generates an inhomogeneous Poisson process of immigrants \(\mathcal I\)
    on \((t_0,T]\) with intensity \(\lambda_0(t)=g\{\mu+x(t)^\top\gamma\}\),
    since the baseline is bounded by \(M<\infty\).

    \medskip
    \noindent
    \textit{2. Direct offspring distribution and lag inversion.}
    For any parent event occurring at time \(t_p\), its direct offspring form a
    Poisson process with lag intensity \(\phi(d)=\rho\beta e^{-\beta d}\) for \(d>0\).
    An offspring at time \(t_e=t_p+d\) falls in the forecast window \((t_0,T]\) if
    and only if \(t_0 < t_p+d \leq T\), which is equivalent to
    \(d\in(a(t_p),b(t_p)]\) with
    \(a(t_p)=\max(0,t_0-t_p)\) and \(b(t_p)=T-t_p\).
    When \(b(t_p)>a(t_p)\), the expected number of such offspring is
    \[
        \int_{a(t_p)}^{b(t_p)}\rho\beta e^{-\beta d}\,\mathrm{d}d
        =
        \rho\{e^{-\beta a(t_p)}-e^{-\beta b(t_p)}\}
        =
        m(t_p).
    \]
    By the properties of Poisson processes, the number of offspring in this window
    is distributed as \(\operatorname{Poisson}(m(t_p))\). Conditional on this
    count, the lags are independent and identically distributed with density
    proportional to \(\beta e^{-\beta d}\) on \((a(t_p),b(t_p)]\). Inverting the
    corresponding conditional cumulative distribution function yields the exact
    sampling formula \eqref{eq:supp-truncated-lag} implemented in Step~2 via the
    inverse transform method.

    \medskip
    \noindent
    \textit{3. Conditioning on \(\mathcal H_0\) and cluster superposition.}
    The cluster representation of \citet{Hawkes2} establishes that a linear Hawkes
    process is equivalent to the superposition of an immigrant process and
    independent branching trees of offspring generated by each event.
    Conditional on the observed history \(\mathcal H_0=\{t_i:t_i< t_0\}\), any
    descendant that occurred prior to \(t_0\) is already recorded and included in
    \(\mathcal H_0\).
    By the independent-increments property of Poisson processes, the subsequent
    offspring generated on \((t_0,T]\) by any observed event \(t_i\in\mathcal H_0\)
    depend only on its timestamp \(t_i\) and are conditionally independent of the
    portions of their offspring processes realised before \(t_0\).
    Initialising the queue with \(\mathcal P\leftarrow\operatorname{queue}(\mathcal H_0\cup\mathcal I)\)
    and restricting offspring to \((t_0,T]\) therefore correctly accounts for all
    future descendants of both past events and future immigrants without double-counting.
    By superposition, the conditional intensity at any \(t\in(t_0,T]\) is
    \[
        \lambda_0(t)+\sum_{t_i<t}\rho\beta e^{-\beta(t-t_i)}
        =
        \lambda^\star(t),
    \]
    which matches \eqref{eq:supp-intensity}.

    \medskip
    \noindent
    \textit{4. Termination of the algorithm.}
    The generation of descendants corresponds to a branching process with
    immigration, where each parent produces an expected number of offspring
    \(m(t_p)\leq\rho<1\) on \((t_0,T]\).
    Because the branching process is subcritical (\(0\leq\rho<1\)), the total
    number of descendants generated by any ancestor is almost surely finite.
    Moreover, every child occurs strictly later than its parent (\(t_e=t_p+D>t_p\)
    since \(D>0\)). The offspring distribution is truncated to the forecast
    window, so every generated child satisfies \(t_e\leq T\), and when
    \(b(t_p)\leq a(t_p)\), no offspring are generated.
    Hence, the parent queue \(\mathcal P\) empties almost surely after a finite
    number of operations, guaranteeing that Algorithm~\ref{alg:supp-cluster}
    terminates in finite steps.
\end{proof}

\paragraph{Remark: piecewise-constant covariates.}

In the present model, the piecewise-constant covariate representation permits
simple refinements of both methods.
Let
\[
    I_k=(\tau_{k-1},\tau_k]\cap(t_0,T],
    \qquad
    \lambda_0(t)=\lambda_{0k}\quad\text{on }I_k.
\]
For piecewise-constant covariates, the modified-thinning implementation may
use the exact local bound
\[
    M_k=\lambda_{0k}
    \qquad\text{on }I_k,
\]
provided that covariate-grid boundaries are treated as deterministic
state-update times. If the next boundary occurs before the proposed candidate
time, the algorithm advances to that boundary, updates the baseline bound, and
generates a new proposal.

For the cluster method, this means that immigrant thinning can be avoided altogether.
For every non-empty interval \(I_k\), the number of immigrants may be
generated directly as
\[
    N_k^{(0)}
    \sim
    \operatorname{Poisson}\{|I_k|\lambda_{0k}\},
\]
and, conditional on \(N_k^{(0)}\),
\[
    T_{k,j}^{(0)}
    \overset{\mathrm{iid}}{\sim}
    \mathcal U(I_k),
    \qquad j=1,\ldots,N_k^{(0)}.
\]
Both refinements preserve the target distribution while improving computational
efficiency.

\subsection{Computational cost and choice of algorithm}
\label{sec:supp-efficiency}

Algorithms~\ref{alg:supp-ogata} and~\ref{alg:supp-cluster} solve the same
simulation problem. Given the common input
\[
    (\mathcal H_0,x,\mu,\gamma,\rho,\beta,g,t_0,T),
\]
both algorithms ultimately return a chronologically ordered list of $N_{\mathrm F} = |\mathcal E|$ future events,
\[
    \mathcal E=(t_1,\ldots,t_{N_{\mathrm F}}),
    \qquad
    t_0<t_1<\cdots<t_{N_{\mathrm F}}\leq T,
\]
whose conditional distribution, given $(\mathcal H_0,x)$, is that of the
Hawkes process with intensity \eqref{eq:supp-intensity}. The two algorithms
therefore target the same conditional distribution, although they do not
generally return the same realised event configuration. Chronological ordering
is treated as part of the common required output, so the final sorting step of
the cluster algorithm is included in its cost.

The comparison below concerns only the generation of future event times.
To make this complexity analysis self-contained, we restate the notation for the relevant event counts:
\[
    H=|\mathcal H_0|,
    \qquad
    N_{\mathrm F}=|\mathcal E|,
    \qquad
    N_{\mathrm I}=|\mathcal I|.
\]
Thus, \(H\) is the number of observed historical events in $\mathcal H_0$,
\(N_{\mathrm F}\) is the number of returned future events in $(t_0,T]$ (already introduced above as $|\mathcal E|$), and
\(N_{\mathrm I}\) is the number of future immigrants.
Let $P_{\mathrm O}$ denote the number of candidate times within
$(t_0,T]$ examined by modified thinning, and let $P_{\mathrm I}$ denote the number
of within-window candidate immigrants examined when Step~1 of the cluster
algorithm is implemented by thinning. When immigrants are generated directly on
piecewise-constant intervals, no immigrant proposal count is required; the
corresponding work is instead proportional to the realised number $N_{\mathrm I}$ of immigrants.

\paragraph{Common computational target and output size.}

Because the two algorithms have the same conditional target distribution,
$N_{\mathrm F}$ has the same distribution under both algorithms. In
particular,
\[
    \bar N_{\mathrm F}
    =
    \mathbb E[N_{\mathrm F}\mid\mathcal H_0,x]
\]
is common to the two methods. Under $0\leq\rho<1$, it satisfies
\begin{equation}
    \label{eq:supp-expected-events}
    \bar N_{\mathrm F}
    \leq
    \frac{
        \Lambda_0(t_0,T)
        +
        A_{\mathrm H}
    }{1-\rho},
\end{equation}
where
\[
    \Lambda_0(t_0,T)
    =
    \int_{t_0}^T\lambda_0(t)\,\mathrm dt
\]
and
\begin{equation}
    \label{eq:supp-history-yield}
    A_{\mathrm H}
    =
    \sum_{t_i\in\mathcal H_0}m(t_i)
    =
    \frac{E_0}{\beta}
    \left\{
    1-\exp[-\beta(T-t_0)]
    \right\}
\end{equation}
is the expected number of direct future offspring generated on $(t_0,T]$ by
the observed history $\mathcal H_0$, with $m(t_i)$ defined in
\eqref{eq:supp-truncated-mean}. The identity in \eqref{eq:supp-history-yield}
follows from substituting $a(t_i)=t_0-t_i$ and $b(t_i)=T-t_i$ into
\eqref{eq:supp-truncated-mean} and using the definition of $E_0$. To derive
\eqref{eq:supp-expected-events}, consider the cluster construction. Every
future event is either an immigrant or a direct offspring of an observed or
future event. Taking conditional expectations therefore gives
\begin{align*}
    \bar N_{\mathrm F}
    ={}    &
    \Lambda_0(t_0,T)
    +
    A_{\mathrm H}
    +
    \mathbb E\left[
        \sum_{t_j\in\mathcal E}m(t_j)
        \mathrel{\Big|}\mathcal H_0,x
    \right]  \\
    \leq{} &
    \Lambda_0(t_0,T)
    +
    A_{\mathrm H}
    +
    \rho\bar N_{\mathrm F},
\end{align*}
because $m(t_j)\leq\rho$ for every $t_j\in(t_0,T]$. Rearranging this
inequality yields \eqref{eq:supp-expected-events}. Moreover, immigrants and
direct offspring of the observed history are subsets of the future event set, so
\[
    \bar N_{\mathrm F}
    \geq
    \Lambda_0(t_0,T)+A_{\mathrm H}.
\]
Hence,
\begin{equation}
    \label{eq:supp-expected-events-bounds}
    \Lambda_0(t_0,T)+A_{\mathrm H}
    \leq
    \bar N_{\mathrm F}
    \leq
    \frac{\Lambda_0(t_0,T)+A_{\mathrm H}}{1-\rho}.
\end{equation}
Because the common required output is an explicit list of
\(N_{\mathrm F}\) event times, every algorithm must generate and store each
event at least once. Consequently, generating the output requires at least
\(O(N_{\mathrm F})\) operations on each sample path, and at least
\(O(\bar N_{\mathrm F})\) operations in expectation.

\paragraph{Cost of modified thinning.}

For the exponential triggering kernel, the observed history contributes to
future simulation only through the scalar \(E_0\) defined in
\eqref{eq:supp-initial-excitation}.
Computing \(E_0\) directly from \(\mathcal H_0\) requires \(O(H)\)
operations. If \(E_0\) is already available (for example, from a preceding
likelihood or posterior computation), the initialisation cost is \(O(1)\).

After initialisation, each candidate proposal requires a constant number of
elementary operations and random variates. The realised modified-thinning cost is
therefore $O(H+P_{\mathrm O})$ when $E_0$ must be computed, and $O(P_{\mathrm O})$
when it is already available. Since every future event corresponds to one accepted proposal,
\[
    P_{\mathrm O}\geq N_{\mathrm F}.
\]
Provided $\bar N_{\mathrm F}>0$, define the average acceptance efficiency by
\begin{equation}
    \label{eq:supp-average-acceptance}
    \bar p_{\mathrm{acc}}
    =
    \frac{
        \mathbb E[N_{\mathrm F}\mid\mathcal H_0,x]
    }{
        \mathbb E[P_{\mathrm O}\mid\mathcal H_0,x]
    }.
\end{equation}
The expected proposal count is therefore
\begin{equation}
    \label{eq:supp-thinning-expected-candidates}
    \mathbb E[P_{\mathrm O}\mid\mathcal H_0,x]
    =
    \frac{\bar N_{\mathrm F}}{\bar p_{\mathrm{acc}}}.
\end{equation}
Consequently, the expected total cost of modified thinning is
\[
    O\left(
    H
    +
    \frac{\bar N_{\mathrm F}}{\bar p_{\mathrm{acc}}}
    \right)
\]
when $E_0$ is computed from $\mathcal H_0$, and $O(\bar N_{\mathrm F}/\bar p_{\mathrm{acc}})$ when $E_0$ is already available.
If \(\bar N_{\mathrm F}=0\), the proposal cost is instead described directly
by \(\mathbb E[P_{\mathrm O}\mid\mathcal H_0,x]\), which may remain positive if
the dominating rate is not identically zero.
Accepted events are generated in chronological order. Hence, no subsequent
sorting is required. The algorithm uses $O(1)$ algorithmic working memory,
excluding storage of the covariate trajectory and returned event list, because
its dynamic state consists only of the current time and the scalar excitation
state. Including the returned event times, its total storage is $O(N_{\mathrm F})$.

The algorithm-specific overhead is therefore the number $P_{\mathrm
  O}-N_{\mathrm F}$ of rejected proposals. At a candidate time
$s=t+\Delta>t$ proposed from the current state $(t,E)$, the local
rejection probability decomposes into two additive components:
\begin{equation}
    \label{eq:supp-ogata-rejection}
    1-p_{\mathrm{acc}}(s\mid t,E)
    =
    \frac{M-\lambda_0(s)}{M+E}
    +
    \frac{E\{1-\exp[-\beta(s-t)]\}}{M+E}.
\end{equation}
The first term arises from the gap $M-\lambda_0(s)$ between the dominating
baseline bound and the true baseline intensity, whereas the second term is the
contribution to the rejection probability caused by excitation decay between $t$ and $s$.

Conditional on the current excitation $E$, the candidate waiting time is $\Delta=s-t\sim\operatorname{Exp}(M+E)$, which satisfies
\[
    \mathbb E\left[
        1-\exp(-\beta\Delta)
        \mid E
        \right]
    =
    \frac{\beta}{M+E+\beta}.
\]
We average the excitation-decay term over
\(\Delta\sim\operatorname{Exp}(M+E)\). Its expected contribution to the
rejection probability is
\begin{equation}
    \label{eq:supp-excitation-decay-diagnostic}
    D(E)
    =
    \mathbb E\left[
        \frac{E\{1-\exp(-\beta\Delta)\}}{M+E}
        \mathrel{\Big|}E
        \right]
    =
    \frac{E}{M+E}
    \frac{\beta}{M+E+\beta}.
\end{equation}
Thus, $D(E)$ factorises as the excitation share $E/(M+E)$ in the
dominating rate multiplied by the expected fractional decay
$\beta/(M+E+\beta)$ of that excitation over the candidate waiting
interval.  Consequently, excitation decay creates substantial
rejection only when the excitation forms an appreciable fraction of
the dominating rate and decays significantly before the next proposal.
If $\lambda_0$ is piecewise constant and covariate-grid boundaries are
treated as deterministic state updates as described above, the exact
local bound $M_k=\lambda_{0k}$ is used on $I_k$, eliminating baseline
overestimation within each interval.

\paragraph{Cost of the cluster representation.}

For Algorithm~\ref{alg:supp-cluster} as stated, all $H$ observed events are
inserted into the parent queue and processed once. Each future event is also
generated, inserted into the output, and processed as a parent exactly once.
The descendant-generation stage therefore requires
\[
    O(H+N_{\mathrm F})
\]
operations. It has no rejection overhead because all offspring sampled from
the truncated offspring distribution fall in the forecast window.

The immigrant-generation stage has cost $O(P_{\mathrm I})$. Under thinning
with the global bound $M$,
\[
    \mathbb E[P_{\mathrm I}\mid x]
    =
    M(T-t_0).
\]
Provided $\Lambda_0(t_0,T)>0$, define the average immigrant acceptance
efficiency by
\begin{equation}
    \label{eq:supp-immigrant-efficiency}
    \bar p_0
    =
    \frac{
        \mathbb E[N_{\mathrm I}\mid x]
    }{
        \mathbb E[P_{\mathrm I}\mid x]
    }
    =
    \frac{\Lambda_0(t_0,T)}{M(T-t_0)}.
\end{equation}
Thus,
\begin{equation}
    \label{eq:supp-cluster-expected-candidates}
    \mathbb E[P_{\mathrm I}\mid x]
    =
    \frac{\Lambda_0(t_0,T)}{\bar p_0}.
\end{equation}
If $\Lambda_0(t_0,T)=0$, the baseline vanishes almost everywhere on
$(t_0,T]$ and the immigrant-generation stage may be omitted.
For piecewise-constant covariates, immigrants can instead be generated
directly on each non-empty interval $I_k$. In that case, there are no rejected
immigrant proposals, and the immigrant-generation cost is proportional to the
realised number $N_{\mathrm I}$ of immigrants.

Unlike modified thinning, the basic queue implementation does not generate all
events in chronological order. If the common output is required to be
chronologically ordered, sorting adds
\[
    O(N_{\mathrm F}\log N_{\mathrm F})
\]
operations. Thus, conditional on the realised counts, the total cost of the
cluster algorithm as stated is
\[
    O\left(
    H+P_{\mathrm I}+N_{\mathrm F}
    +N_{\mathrm F}\log N_{\mathrm F}
    \right).
\]
When direct immigrant generation is available, this becomes
\[
    O\left(
    H+N_{\mathrm F}
    +N_{\mathrm F}\log N_{\mathrm F}
    \right),
\]
because $N_{\mathrm I}\leq N_{\mathrm F}$.

Total storage is bounded by $O(H+N_{\mathrm F})$. Excluding the common output,
the working memory can still depend on the number of parents awaiting processing.

\paragraph{Like-for-like comparison.}

Collecting the preceding results, the realised costs are
\begin{align*}
    C_{\mathrm{thin}}
     & =
    O\left(H+P_{\mathrm O}\right),
    \\
    C_{\mathrm{cluster}}
     & =
    O\left(
    H+P_{\mathrm I}+N_{\mathrm F}
    +N_{\mathrm F}\log N_{\mathrm F}
    \right),
\end{align*}
where $C_{\mathrm{thin}}$ reduces to $O(P_{\mathrm O})$ when $E_0$ is already available.
The corresponding expected costs satisfy
\begin{align}
    \mathbb E[C_{\mathrm{thin}}\mid\mathcal H_0,x]
     & =
    O\left(
    H
    +
    \frac{\bar N_{\mathrm F}}{\bar p_{\mathrm{acc}}}
    \right),
    \label{eq:supp-expected-cost-thinning}
    \\
    \mathbb E[C_{\mathrm{cluster}}\mid\mathcal H_0,x]
     & =
    O\left(
    H
    +
    \frac{\Lambda_0(t_0,T)}{\bar p_0}
    +
    \bar N_{\mathrm F}
    +
    \mathbb E[N_{\mathrm F}\log N_{\mathrm F}\mid\mathcal H_0,x]
    \right),
    \label{eq:supp-expected-cost-cluster}
\end{align}
where the first expectation reduces to $O(\bar N_{\mathrm F}/\bar p_{\mathrm{acc}})$ when $E_0$ is already available.
When immigrants are generated directly, the expected cluster cost reduces to
\[
    O\left(
    H
    +
    \bar N_{\mathrm F}
    +
    \mathbb E[N_{\mathrm F}\log N_{\mathrm F}\mid\mathcal H_0,x]
    \right),
\]
because $\mathbb E[N_{\mathrm I}\mid x]=\Lambda_0(t_0,T)\leq\bar
N_{\mathrm F}$.  Both methods incur the unavoidable $O(N_{\mathrm F})$
time required to materialise the common output and $O(N_{\mathrm F})$
storage if the complete event list is retained.  The term $P_{\mathrm
  O}-N_{\mathrm F}$ measures the rejection overhead of modified
thinning, while $P_{\mathrm I}-N_{\mathrm I}$ measures the
immigrant-thinning overhead of the cluster algorithm. These
expressions compare the algorithms for the same conditioning
information, the same forecast window, and the same required output.

\paragraph{Regimes of poor performance.}

The preceding complexity expressions do not provide a universal ranking of the
two algorithms, because they suppress implementation-dependent constants and
depend on random output sizes. To obtain guidelines on when to use one or the
other, we instead study the regimes in which either algorithm is expected to
perform poorly.

The analysis identifies four sources of overhead: rejected proposals in
modified thinning, rejected immigrant proposals in the cluster algorithm,
individual processing of a long observed history, and chronological sorting
of a large future event set.

\medskip
\noindent
\textit{1. When modified thinning is expected to perform poorly.}
Modified thinning has expected proposal cost $\bar N_{\mathrm F}/\bar p_{\mathrm{acc}}$.
It is therefore inefficient whenever $\bar p_{\mathrm{acc}}\ll 1$. If
\(\bar p_{\mathrm{acc}}\leq\varepsilon\), then the ratio of the expected
number of proposals to the expected number of returned events is at least
\(\varepsilon^{-1}\). For example,
\(\bar p_{\mathrm{acc}}=0.1\) corresponds to ten expected proposals per
expected returned event.
Equation~\eqref{eq:supp-ogata-rejection} identifies two distinct regimes that
can make $\bar p_{\mathrm{acc}}$ small:
\begin{itemize}
    \item \textit{Loose baseline domination:} When $M\gg\lambda_0(s)$ across most
          of $(t_0,T]$, the gap $M-\lambda_0(s)$ between the upper bound and the
          actual baseline lowers proposal acceptance. The immigrant efficiency
          $\bar p_0$ defined in \eqref{eq:supp-immigrant-efficiency} measures the
          average tightness of the global baseline bound when $\Lambda_0(t_0,T)>0$.
          In the limiting baseline-only case $E(t)=0$, and when the same global bound
          $M$ is used throughout the forecast window, the ratio of expected accepted
          events to expected within-window proposals is $\bar p_0$. More generally, when the
          excitation component remains weak, one expects $\bar p_{\mathrm{acc}}$ to
          be close to $\bar p_0$. Consequently, $\bar p_0\ll 1$ indicates substantial
          baseline-proposal inefficiency; this becomes computationally important
          when the baseline proposal volume $M(T-t_0)=\Lambda_0(t_0,T)/\bar p_0$ is also
          appreciable. This baseline-bound overhead is removed by the
          interval-specific refinement described above.
    \item \textit{Rapid excitation decay between proposals.}
          The quantity
          \[
              D_0
              =
              \frac{E_0}{M+E_0}
              \frac{\beta}{M+E_0+\beta}
          \]
          measures the initial contribution of excitation decay to rejection.
          A large \(D_0\) warns that modified thinning may initially waste
          proposals, although it does not bound the overall acceptance efficiency
          because the excitation state evolves during simulation.
\end{itemize}
Modified thinning is therefore especially vulnerable when $\bar p_0\ll 1$ with
appreciable baseline proposal volume $M(T-t_0)$, or when $D_0$ is large.

\medskip
\noindent
\textit{2. When the cluster algorithm is expected to perform poorly.}
Equation~\eqref{eq:supp-expected-cost-cluster} identifies three distinct
problematic regimes for the cluster algorithm:
\begin{itemize}
    \item \textit{Inefficient immigrant thinning:} When immigrants are generated
          by global thinning, the expected number of immigrant proposals is
          $\mathbb E[P_{\mathrm I}\mid x]=M(T-t_0)$. When $\Lambda_0(t_0,T)>0$, this
          satisfies $\mathbb E[P_{\mathrm I}\mid x]=\Lambda_0(t_0,T)/\bar p_0$ by
          \eqref{eq:supp-immigrant-efficiency}. If $\bar p_0\leq\varepsilon$, then the ratio
          of the expected number of immigrant proposals to the expected number of
          returned immigrants is at least $\varepsilon^{-1}$ (for example,
          $\bar p_0=0.05$ corresponds to twenty expected proposals per expected
          immigrant). Immigrant thinning may be costly when $\bar p_0\ll 1$ and the
          baseline proposal volume $M(T-t_0)=\Lambda_0(t_0,T)/\bar p_0$ is large.
          This immigrant-thinning overhead is absent under the direct
          interval-specific generation described above.
    \item \textit{A long observed history with little future influence:} The
          cluster algorithm as stated processes every historical event individually,
          costing $O(H)$ operations. By \eqref{eq:supp-history-yield}, the total
          expected direct future offspring generated by $\mathcal H_0$ is
          $A_{\mathrm H}$. The average expected future offspring yield per
          historical event is $A_{\mathrm H}/H$, for $H>0$. A small value indicates
          that most of the $O(H)$ history-processing cost is associated with events
          having little remaining influence on the forecast window (as occurs when
          historical events occurred far in the past, $\beta$ is large, or the
          forecast horizon $T-t_0$ is short).
          Modified thinning avoids processing the historical events individually during
          the simulation because their joint contribution is represented by the scalar
          \(E_0\). If \(E_0\) is already available, it also avoids the initial
          \(O(H)\) pass through the history.
          A conservative, readily computable warning condition is
          \[
              H \gg \Lambda_0(t_0,T)+A_{\mathrm H},
              \qquad\text{with}\qquad
              \frac{A_{\mathrm H}}{H}\ll 1,
          \]
          indicating that many historical records are processed despite having
          negligible remaining direct influence.
    \item \textit{Chronological sorting of large event sets:} The basic queue
          implementation generates events in genealogical tree order rather than
          chronological order. Returning chronologically ordered output adds a sorting
          cost of $O(N_{\mathrm F}\log N_{\mathrm F})$ operations. With the
          convention $0\log 0 = 0$, convexity of $z\mapsto z\log z$ gives
          \[
              \mathbb E[N_{\mathrm F}\log N_{\mathrm F}\mid\mathcal H_0,x]
              \geq
              \bar N_{\mathrm F}\log\bar N_{\mathrm F}.
          \]
          Therefore, when $\Lambda_0(t_0,T)+A_{\mathrm H}$ is substantially larger
          than one, the lower bound in \eqref{eq:supp-expected-events-bounds} implies
          a correspondingly large lower bound on the expected sorting cost. A large
          value of
          \[
              \frac{\Lambda_0(t_0,T)+A_{\mathrm H}}{1-\rho}
          \]
          provides a more conservative warning that branching may further increase
          the event volume, but it does not by itself imply that the expected event
          count is large on the finite forecast window. In this large-event-volume
          regime, sorting overhead and parent-queue storage become significant
          disadvantages of the cluster representation.
\end{itemize}

The table below
summarises the preceding analysis.
\begin{center}
    \small
    \renewcommand{\arraystretch}{1.25}
    \begin{tabularx}{\textwidth}{
        >{\raggedright\arraybackslash}p{0.24\textwidth}
        >{\raggedright\arraybackslash}X
        >{\raggedright\arraybackslash}p{0.25\textwidth}
        }
        \hline
        Observable condition                                                                             & Inefficiency mechanism                                               & Vulnerable method                          \\
        \hline
        $\bar p_0 \ll 1$ and $M(T-t_0)$ large                                                            & Large baseline-driven proposal volume caused by a loose global bound & Both methods under a global baseline bound \\
        $D_0$ large                                                                                      & Significant initial excitation-bound decay                           & Modified thinning                          \\
        $H$ large and $A_{\mathrm H}/H \ll 1$                                                            & Many historical events with low future yield                         & Cluster algorithm as stated                \\
        $\Lambda_0(t_0,T)+A_{\mathrm H}$ large                                                           & Large lower bound on the expected future event count                 & Cluster sorting and queue storage          \\
        $\rho$ close to one, appreciable remaining horizon, and non-negligible $\Lambda_0+A_{\mathrm H}$ & Potentially substantial additional branching                         & Both methods, especially cluster sorting   \\
        \hline
    \end{tabularx}
\end{center}

These diagnostics identify potential failure regimes rather than exact
crossover points. Modified thinning is expected to be attractive when tight local
bounds are available, when $E_0$ is already known, or when events must be
generated in chronological order. Excluding the covariate trajectory and
returned event list, it requires only $O(1)$ algorithmic working memory.

The cluster algorithm is expected to be attractive when immigrants can be
generated directly, when the observed history is not disproportionately long
relative to its remaining influence, or when chronological ordering is
unnecessary or inexpensive.


\section{Convergence Diagnostics}
\label{diag}
This appendix reports convergence diagnostics for the Markov Chain
Monte Carlo (MCMC) simulations used to estimate the parameters of the
model (\ref{condIntHawkesInh}), applied to spore release data in
Section \ref{scab}. All posterior estimates and MCMC simulations were
performed using RStan, the R interface to Stan. Convergence is
assessed using traceplots for the key parameters of the model, as well
as standard numerical diagnostics, including the Gelman–Rubin
statistic ($\hat{R}$) and the effective sample size (ESS).

\subsection{Traceplots}

Figure \ref{fig:diag} reports traceplots for the main parameters of
the model. The chains mix well, and no visible trends are observed,
suggesting satisfactory convergence.
\begin{figure}
    \centering
    \includegraphics[width=\linewidth]{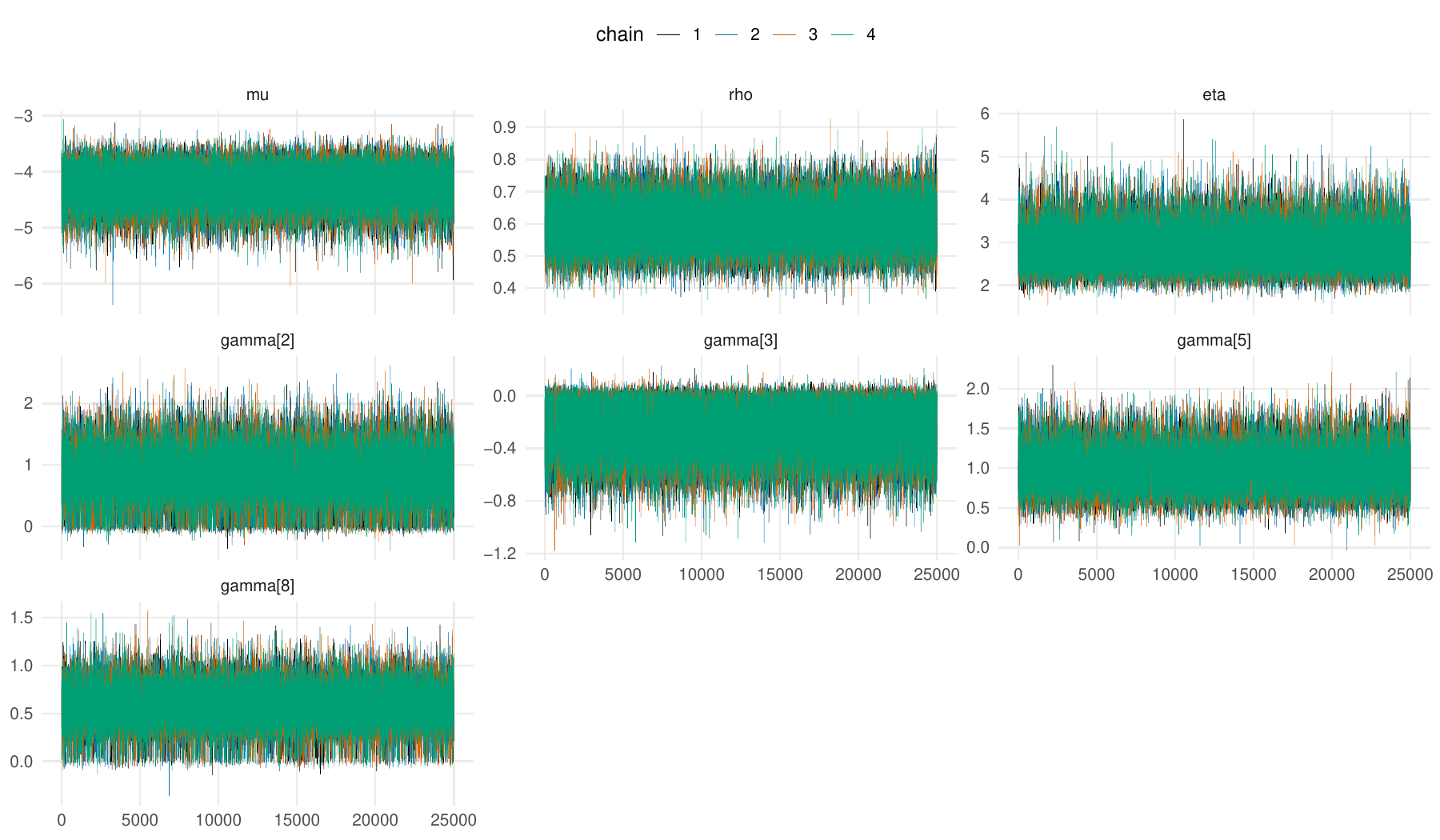}

    \caption{Traceplots for  the constant $\mu$, the parameters $\rho$ and $\eta=\frac{1}{\beta}$ of
        the excitation function and  the regression coefficients
        $\gamma_2$,$\gamma_3$,$\gamma_5$ and $\gamma_8$ of the selected covariates.}
    \label{fig:diag}
\end{figure}

\subsection{Convergence Statistics}
Table \ref{tab:convergence} reports Gelman–Rubin statistics and effective sample sizes
for the parameters of interest. All $\hat{R}$ statistics are close to 1, indicating
convergence of the Markov chains.

\begin{table}[ht]
    \centering
    \caption{Convergence diagnostics for the selected parameters. }
    \label{tab:convergence}
    \begin{tabular}{lcc}
        \hline
        {Parameter} & {$\hat{R}$} & {ESS}  \\
        \hline
        $\mu$       & 1.00        & 50,662 \\
        $\eta$      & 1.00        & 73,988 \\
        $\rho$      & 1.00        & 69,810 \\
        $\gamma_2$  & 1.00        & 13,359 \\
        $\gamma_3$  & 1.00        & 15,187 \\
        $\gamma_5$  & 1.00        & 16,542 \\
        $\gamma_8$  & 1.00        & 27,049 \\
        \hline
    \end{tabular}
\end{table}

\bibliographystyle{apalike-jabes}

\newpage
\bibliography{references}

\end{document}